\documentclass[a4paper,12pt]{article}

\usepackage[
  a4paper,
  top=2.0cm,
  bottom=2.0cm,
  left=2.0cm,
  right=2.0cm,
  headheight=15pt,
  headsep=0.8cm,
  footskip=1.2cm
]{geometry}
\usepackage{setspace}
\usepackage{lineno}
\usepackage{amsmath,amssymb,amsfonts}
\usepackage{amsthm}
\usepackage{thmtools}
\usepackage{mathtools}
\usepackage{bm}
\usepackage[title]{appendix}
\usepackage{textcomp}
\usepackage{manyfoot}
\usepackage{booktabs}
\usepackage{listings}
\usepackage{hyperref}
\usepackage[capitalise,nameinlink,noabbrev]{cleveref}
\usepackage{xcolor}
\usepackage{graphicx}
\usepackage{float}
\usepackage{indentfirst}
\usepackage{enumerate}
\usepackage{stmaryrd}
\usepackage{multicol}
\usepackage{multirow}
\usepackage{array}
\usepackage{subcaption}
\usepackage{rotating}
\usepackage{microtype}
\usepackage{algorithm}
\usepackage{algorithmicx}
\usepackage{algpseudocode}
\usepackage[numbers]{natbib}
\hypersetup{
  colorlinks=false,
  hidelinks,
  hypertexnames=false
}
\newcommand{\thickhline}{\specialrule{0.8pt}{0pt}{0pt}}
\makeatletter
\renewcommand{\theHalgorithm}{\arabic{algorithm}}
\renewcommand{\theHALG@line}{\theHalgorithm.\arabic{ALG@line}}
\makeatother
\theoremstyle{plain}
\newtheorem{theorem}{Theorem}
\newtheorem{proposition}[theorem]{Proposition}
\newtheorem{lemma}[theorem]{Lemma}
\newtheorem{corollary}[theorem]{Corollary}
\theoremstyle{definition}
\newtheorem{definition}{Definition}
\theoremstyle{remark}

\providecommand{\bibcommenthead}{}

\newcommand{\FUNDING}{This work was supported by the Brazilian research agencies CAPES, CNPq, and FAPESP, as well as the French research agency ANR.}
\title{An Angle-Based Algorithmic Framework for the Interval Discretizable Distance Geometry Problem\thanks{\FUNDING}}
\author{Wagner da Rocha \and Carlile Lavor \and Leo Liberti \and Let\'icia de Melo Costa \and Leonardo D. Secchin \and Th\'er\`ese E. Malliavin}
\date{\today}

\begin{document}
\maketitle

	\begin{abstract}
		Distance Geometry is central to protein structure determination from Nuclear Magnetic Resonance data, where distance restraints are naturally uncertain. We study the interval Discretizable Distance Geometry Problem (\textit{i}DDGP) and introduce two angular Branch-and-Prune (BP) frameworks: the interval Angular Branch-and-Prune (\textit{i}ABP) method and its torsion-guided extension, the interval Torsion-angle Branch-and-Prune (\textit{i}TBP) method. Both methods transform interval distance information into angular constraints on circular arcs; \textit{i}ABP uses these constraints to reduce explicit feasibility checks during branching, while \textit{i}TBP further incorporates prescribed torsion-angle intervals to encode local chirality and planarity information. We develop the corresponding geometric foundations and describe a systematic construction of biologically meaningful \textit{i}DDGP instances from Protein Data Bank structures. Computational experiments indicate that \textit{i}ABP improves the embedding-normalized feasible-solution yield relative to interval BP (\textit{i}BP), a standard enumerative BP baseline for \textit{i}DDGP instances, and that both angular methods find feasible realizations for more instances than \textit{i}BP. Moreover, \textit{i}TBP attains the lowest minimum RMSD on all instances for which both angular methods find feasible realizations, highlighting the benefit of incorporating torsion-angle information.
	\end{abstract}
	
	\noindent \textbf{keywords:} Distance Geometry, Branch-and-Prune, Nuclear Magnetic Resonance, Protein Structures Characterization
	
	\section{Introduction}\label{sec:1}
	
	Determining the three-dimensional structure of proteins is a central problem in structural biology, as protein function is intrinsically determined by molecular conformation~\cite{2011_donald_aismb}. Given a protein's amino acid sequence, its spatial structure can be inferred through a variety of experimental techniques~\cite{1997_brunger_xcanrcvosad, 2019_mitra_rev}. Among these, Nuclear Magnetic Resonance (NMR) spectroscopy is particularly valuable for structure determination in solution. NMR experiments yield interatomic distance restraints, primarily between hydrogen atoms, derived from the Nuclear Overhauser Effect (NOE), which arises from magnetization transfer via dipolar coupling. Although these distance estimates are affected by uncertainties stemming from magnetization pathways and molecular flexibility, they provide essential geometric information for the reconstruction of protein conformations~\cite{1998_guntert_scobmfnd, 2008_levitt_sdbonmr}.
	
	The task of computing a protein's three-dimensional conformation from such distance constraints is formalized as the Molecular Distance Geometry Problem (MDGP), a central application of Distance Geometry~\cite{2013_mucherino_dgtmaa}. Given a set of pairwise inter-atomic distances, the objective is to reconstruct the corresponding molecular structure in Euclidean space~\cite{2014_liberti_edgaa}. This problem is known to be NP-hard~\cite{1980_saxe_eowgikisnh}, making efficient algorithmic approaches essential for practical applications.
	
	The most general formulation of this class of problems is the interval Distance Geometry Problem (\textit{i}DGP), which involves embedding a weighted graph into a normed space $\mathbb{R}^K$, with $K \in \mathbb{N}$. Given a graph $G = (V, E, \mathcal{D})$, the goal is to find a function ${x} \ : \ V \rightarrow \mathbb{R}^K$ that maps each vertex to a point in $\mathbb{R}^K$ such that the distances between mapped points match the given values in $\mathcal{D}$. In the case of protein structure determination, the norm is Euclidean, and the dimensionality is $K = 3$, leading to the \textit{i}DGP in $\mathbb{R}^3$, formally described as follows~ \cite{2021_lavor_itspitibpaftdmdgp}:
	\begin{definition}
		Let $G = (V, E, \mathcal{D})$ be a simple, weighted, undirected graph, where $V$ is its set of vertices, $E$ is its set of edges, and $\mathcal{D}: E \to \mathcal{I}(\mathbb{R}_+)$ is a function that assigns an interval from $\mathcal{I}(\mathbb{R}_+)$, a set of positive real intervals, to each edge. The interval Distance Geometry Problem in $\mathbb{R}^3$ consists of finding a function ${x} \ : \ V \rightarrow \mathbb{R}^3$, called the realization, such that:
		\begin{equation}
			\forall \{u, v\} \in E, \ \|{x}_u - {x}_v\| \in \mathcal{D}_{u,v} \equiv \big[\underline{d}_{u,v}, \overline{d}_{u,v}\big],
			\label{sec:1:eq:instance_iDGP}
		\end{equation}
		
		\noindent with ${x}_u = {x}(u),\ {x}_v = {x}(v),\ \mathcal{D}_{u,v} = \mathcal{D}\big(\{u, v\}\big) \ \forall u, v \in V$, and $\|\cdot\|$ is the Euclidean norm.
	\end{definition}
	
	The survey~\cite{2014_liberti_edgaa} provides a detailed discussion on the existence and uniqueness of solutions to~\eqref{sec:1:eq:instance_iDGP}. In this context, each interval of~\eqref{sec:1:eq:instance_iDGP} is defined as $\mathcal{D}_{u,v} = \{d \in \mathbb{R}_+ \ | \ \underline{d}_{u,v} \leq d \leq \overline{d}_{u,v}\}$ which allows the problem to be rewritten as a system of inequalities:
	\begin{equation*}
		\forall \{u, v\} \in E, \ \underline{d}_{u,v} \leq \|{x}_u - {x}_v\| \leq \overline{d}_{u,v}.
	\end{equation*}
	
	When the distance between two vertices $u, v \in V$ is known exactly, the interval $\mathcal{D}_{u,v}$ reduces to a single value, transforming the corresponding inequality in~\eqref{sec:1:eq:instance_iDGP} into an equation. This means that for certain vertex pairs, the distances are exact.
	
	The case in which all distance constraints are precise is known as the Distance Geometry Problem (DGP). Early approaches to solving the DGP formulated it as a global optimization problem~\cite{2014_liberti_edgaa,1997_more_gcfdgp},
	\begin{equation}
		\min_{X \in \mathbb{R}^{|V| \times 3}} \sum_{\{u,v\}\in E} \Big(\|{x}_u - {x}_v\|^2 - d^2_{u,v}\Big)^2,
		\label{sec:1:eq:minglob_DGP}
	\end{equation}
	
	\noindent where a solution $X := [x_w^t]_{w \in V}$ corresponds to a global minimizer of~\eqref{sec:1:eq:minglob_DGP}. Although conceptually simple, this formulation gives rise to a highly nonconvex objective function with many local minima, making the computation of globally optimal realizations challenging.
	
	Formulation~\eqref{sec:1:eq:minglob_DGP} gives rise to a nonlinear least-squares problem. Classical solution approaches are based on local optimization schemes, including gradient-type methods~\cite{1993_GHH}, trust-region methods, and Levenberg--Marquardt-type algorithms~\cite{1978_more_TLMAIaT,2022_GM}. In distance geometry, this perspective has been widely adopted to generate or refine approximate realizations~\cite{2013_mucherino_dgtmaa,2014_liberti_edgaa}.
	
	Over the last two decades, substantial progress has been achieved by exploiting structural properties of the underlying graph $G$, leading to specialized algorithmic paradigms for the DGP~\cite{2002_dong_altafstmdgpweiad, 2012_lavor_tdmdgp, 2013_alipanahi_DPSfNDCbSDP,2021_goncalves_anaftkdsodgpwed}. Despite these advances, the problem remains challenging from both computational and practical standpoints, particularly in the presence of sparse, noisy, or interval distance data~\cite{2013_souza_smdgpwidd,2014_liberti_edgaa}.
	
	Two broad families of approaches have emerged. Continuous relaxation methods replace the original nonconvex formulation by convex or smooth approximations, while discrete enumerative methods exploit problem structure to achieve exactness.
	
	Among continuous approaches, semidefinite programming (SDP) relaxations of the Euclidean Distance Matrix (EDM) completion problem are prominent~\cite{2005_alhomidan_LAaiA, 1999_alfakih_SEDMCPVSDP, 2013_alipanahi_DPSfNDCbSDP, 2017_drusvyatskiy_NEDRRFRatPF}. These formulations recover low-rank realizations consistent with the given distances but generally offer no guarantee of three-dimensionality and do not encode chirality or torsion angle information. While facial reduction and low-rank reformulations reduce computational costs in the exact-distance setting, extending SDP-based models to interval data remains demanding~\cite{2017_drusvyatskiy_NEDRRFRatPF,2017_huang_LRMCUNNMAFR,2003_burer_ANPAfSSDPvLRF,2021_bellavia_ARIPMfLRSDPPwAtMC}.
	
	In contrast, discrete enumerative approaches exploit the combinatorial structure induced by specific graph classes, most notably the Discretizable DGP (DDGP). In this setting, from the fourth vertex onward each vertex is adjacent to at least three previously placed vertices in a prescribed order, so that its position can be determined via the intersection of spheres~\cite{2015_cassioli_dvoidg}. The resulting search space is discrete and can be encoded as a binary tree. A prominent special case is the Discretizable Molecular DGP (DMDGP), in which each vertex from the fourth onward is adjacent to exactly three consecutive predecessors. This additional structure induces symmetry properties that allow all feasible realizations to be generated from a single solution via partial reflections~\cite{2021_lavor_otoofds,2013_liberti_ctnosokdi,2013_nucci_stdmdgpbmrt}.
	
	Branch-and-Prune (BP) algorithms exploit this discretization by sequentially placing vertices and exploring the finite set of feasible realizations, pruning infeasible branches as soon as violations are detected~\cite{2008_liberti_abpaftmdgp}. For exact distance data, although BP has exponential worst-case complexity, it provides exact feasibility guarantees and explicit enumeration of noncongruent solutions. When distance data are uncertain, interval information must be incorporated, and finite sampling procedures no longer certify infeasibility of the underlying continuous interval problem when the sampled tree is exhausted. The interval Branch-and-Prune (\textit{i}BP) algorithm extends BP by sampling distances within prescribed intervals, transforming the binary search tree into a $2k$-ary tree~\cite{2013_lavor_tibpaftdmdgpwid}. While increased sampling improves accuracy, it also enlarges the search space substantially~\cite{2017_dambrosio_NEMaMfRPGfDG}. Various refinements mitigate this growth, including hybrid enumeration--optimization strategies~\cite{2020_mucherino_mdjeep2} and pruning techniques that reduce uncertainty prior to sampling~\cite{2021_lavor_itspitibpaftdmdgp}.
	
	These developments reveal a fundamental trade-off in distance geometry: enumerative methods preserve exact feasibility and multiplicity of solutions, but require careful control of combinatorial growth under interval uncertainty. The present work advances this pruning-based paradigm by extending angular discretization and torsion-guided strategies to the more general interval DDGP (\textit{i}DDGP).
	
	\medskip
	\noindent \textbf{Contributions:}
	The main contributions of this paper are as follows.
	\begin{enumerate}
		\item Extension of angular discretization techniques previously developed for interval molecular settings (\textit{i}DMDGP)~\cite{2021_lavor_itspitibpaftdmdgp} to the full \textit{i}DDGP framework, thereby generalizing angular Branch-and-Prune strategies beyond DMDGP--type orders~\cite{2008_liberti_abpaftmdgp, 2013_lavor_tibpaftdmdgpwid};
		
		\item Introduction of the \emph{interval} Angular Branch-and-Prune (\textit{i}ABP) algorithm, which discretizes angular intervals rather than distance intervals. Each branching step is interpreted as the intersection of spheres and spherical shells, progressively refining the admissible angular regions induced by the predecessor constraints before sampling, thereby strengthening pruning before sampling;
		
		\item Development of the \emph{interval} Torsion Angle Branch-and-Prune (\textit{i}TBP) algorithm, incorporating torsion angle sign information (e.g., from NMR data). This additional geometric constraint reduces the effective search space, intrinsically enforces local chirality, and enhances the recovery of geometrically faithful conformations;
		
		\item Proposal of a systematic procedure for generating realistic interval instances from Protein Data Bank (PDB) structures. Since PDB data do not inherently follow a discretizable DDGP ordering, a theoretical post-processing step is described to enforce suitable vertex orders. The PDB parser, together with implementations of \textit{i}ABP and \textit{i}TBP, is publicly available.
	\end{enumerate}
	
	\medskip
	\noindent \textbf{Notation:}
	Points in $\mathbb{R}^3$ are denoted by uppercase letters (e.g., $P$) for fixed references and by the lowercase letter $x$ for constructed realizations. Vertices are written as lowercase letters (e.g., $v$), enumerable sets are represented by uppercase letters (e.g., $V$), non-enumerable sets by calligraphic uppercase letters (e.g., $\mathcal{A}$), vectors in boldface (e.g., $\mathbf{v}$), and scalars typically by lowercase Roman or Greek letters (e.g., $d$, $\kappa$). The notation introduced in \Cref{sec:4} follows conventions from structural biology and therefore does not always adhere strictly to this general scheme.
	
	\medskip
	\noindent \textbf{Outline:} \Cref{sec:2} introduces the \textit{i}DDGP and reviews the \textit{i}BP algorithm. \Cref{sec:3} presents interval arc reduction techniques and our proposed methods. \Cref{sec:4} describes the generation of distance constraint instances from protein data. \Cref{sec:5} reports computational results, and \Cref{sec:6} concludes the paper.
	
	\section{The \textit{i}DDGP and the \textit{i}BP Algorithm}\label{sec:2}
	
	To formally define the \textit{i}DDGP, we begin by introducing the following preliminary definitions:
	\begin{definition}
		Let $G = (V, E, \mathcal{D})$ be a simple, weighted, and undirected graph, and let $\prec$ be an ordering on the elements of $V$. We define the following sets of vertices and edges:
		\begin{enumerate}
			\item The set of predecessors of $v \in V$:
			$\gamma(v) = \{ u \in V \mid u \prec v \};$

			\item The set of vertices adjacent to $v \in V$:
			$N(v) = \big\{ u \in V \mid \{u,v\} \in E \big\};$

			\item The set of adjacent predecessors of $v \in V$:
			$U_v = N(v) \cap \gamma(v).$			
			We denote $\kappa_v = |U_v|$;
			
			\item The set of edges with precise distances:
			\[
			E_0 = \big\{ \{u,v\} \in E \ \mid \ \mathcal{D}_{u,v} = [\underline{d}_{u,v}, \overline{d}_{u,v}] \text{ with } \underline{d}_{u,v} = \overline{d}_{u,v} \big\};
			\label{sec:2:eq:E0}
			\]
			
			\item The set of edges with interval distances:
			\[
			E_I = \big\{ \{u,v\} \in E \ \mid \ \mathcal{D}_{u,v} = [\underline{d}_{u,v}, \overline{d}_{u,v}] \text{ with } \underline{d}_{u,v} < \overline{d}_{u,v} \big\}.
			\label{sec:2:eq:EI}
			\]
		\end{enumerate}
	\end{definition}
	
	Considering the definitions above, the \textit{i}DDGP in $\mathbb{R}^3$ is defined as follows~\cite{2021_lavor_itspitibpaftdmdgp}:
	
	\begin{definition}
		\label{sec:2:def:iddgp}
		Let $G = (V, E, \mathcal{D})$ be a simple, weighted, and undirected graph associated with an instance of the \textit{i}DGP in $\mathbb{R}^3$, where $E = E_0 \cup E_I$ and $v_1,\ v_2,\ \dots,\ v_{|V|}$ is a total order on the vertices of $V$. The \textit{i}DDGP consists of all instances of the \textit{i}DGP that satisfy the following assumptions:
		\begin{enumerate}
			\item The subgraph induced by $\{v_1,\ v_2,\ v_3\}$ is a clique in $E_0$, and the associated weights $d_{1,2},\ d_{1,3}$, and $d_{2,3}$ satisfy the strict triangle inequalities;
			
			\item For every vertex $v_i$, with $i = 4,\ \dots,\ |V|$, there exists a subset $\{v_{i_3},\ v_{i_2},\ v_{i_1}\} \subset U_{v_i}$ such that:
			\begin{enumerate}[$1.$]
				\item $\{\{v_i, v_{i_2}\},\ \{v_i, v_{i_1}\}\} \subset E_0$;
				
				\item The subgraph induced by $\{v_{i_3},\ v_{i_2},\ v_{i_1}\}$ is a clique in $E_0$, and the associated weights satisfy the strict triangle inequalities.
			\end{enumerate}
		\end{enumerate}
	\end{definition}
	
	Regarding \cref{sec:2:def:iddgp}, item~1 determines the initial triangle up to Euclidean isometry. For simplicity, we fix this freedom by choosing the reference frame
	\begin{equation*}
		x_1=(d_x-d_{2,3},d_y,0),\qquad
		x_2=(-d_{2,3}, 0 ,0),\qquad
		x_3=(0,0,0),
	\end{equation*}
	
	\noindent where
	\begin{equation*}
		d_x=\frac{d_{1,2}^2+d_{2,3}^2-d_{1,3}^2}{2d_{2,3}},
		\qquad
		d_y=\sqrt{d_{1,2}^2 - d_x^2}.
	\end{equation*}
	
	\noindent Moreover, item~2 establishes local conditions on the predecessor set: subitem~(a) guarantees that, for each vertex \(v_i \in V\), at least two adjacent predecessors are connected to \(v_i\) by precise distances, while subitem~(b) ensures that the selected predecessor triple forms a nondegenerate triangle. This nondegeneracy is required for the angular parametrization used below. In contrast with the exact DDGP, interval distance constraints may still give rise to continuous feasible sets, which are later explored through finite angular samples~\cite{2017_goncalvez_raotidgp}.
	
	In the geometric constructions developed below, we use the following nondegeneracy assumption: the two exact distances used to place each vertex define a nondegenerate circle of intersection. More precisely, for every $i = 4,\dots,|V|$, the distances $d_{i,i_1}$, $d_{i,i_2}$, and $d_{i_1,i_2}$ are assumed to satisfy the strict triangle inequalities, which implies that the radius of the corresponding circle of intersection is positive. Degenerate tangent cases, where the two spheres intersect in a single point, can be treated separately but are not considered in the angular parametrization used here. For the protein instances considered in this work, this assumption is verified in \Cref{sec:4}.
	
	We remark that in the definition of the \textit{i}DDGP, for each vertex $v_i \in V$, where $U_{v_i} = \left\{v_{i_1},\ v_{i_2},\ v_{i_3},\ \dots,\ v_{i_{\kappa_i}}\right\}$, with $i = 4,\ \dots,\ n$, the following nonlinear system in the variable ${x}_i$ arises:
	\begin{equation}
		\left\{\begin{array}{ccccc}
			\underline{d}_{i,i_{\kappa_i}} & \leq & \left\| {x}_i - {x}_{i_{\kappa_i}} \right\| & \leq & \overline{d}_{i,i_{\kappa_i}},\\
			\vdots & \vdots & \multicolumn{1}{c}{\vdots} & \vdots & \vdots\\
			\underline{d}_{i,{i_3}} & \leq & \left\| {x}_i - {x}_{i_3} \right\| & \leq & \overline{d}_{i,{i_3}},\\
			& & \left\| {x}_i - {x}_{i_2} \right\| & = & d_{i,{i_2}},\\
			& & \left\| {x}_i - {x}_{i_1} \right\| & = & d_{i,{i_1}},
		\end{array}\right.
		\label{sec:2:eq:system_xiI_k>3}
	\end{equation}
	
	\noindent with $\kappa_i \geq 3$. When it is necessary to distinguish the input interval from the reduced reachable interval used in the angular construction, we denote the original interval distance by $\mathcal{D}^0$. By squaring both sides of the equalities and inequalities, we observe that each equation represents a sphere and each inequality represents a spherical shell. Thus, finding a solution to system \eqref{sec:2:eq:system_xiI_k>3} is equivalent to the geometric problem of finding the intersection of two spheres $S^i_1 = S\big({x}_{i_1}, d_{i,i_1}\big)$ and $S^i_2 = S\big({x}_{i_2}, d_{i,i_2}\big)$, with $\kappa_i - 2$ spherical shells $\mathcal{S}_3^i = S\big({x}_{i_3}, \mathcal{D}^0_{i,i_3}\big),\ \dots,\ \mathcal{S}_{\kappa_i}^i = S\big({x}_{i_{\kappa_i}}, \mathcal{D}^0_{i,i_{\kappa_i}}\big)$:
	\begin{equation}
		\mathcal{A}_i = S^i_1 \cap S^i_2 \cap \mathcal{S}^i_3 \cap \mathcal{S}^i_4 \cap \dots \cap \mathcal{S}^i_{\kappa_i}.
		\label{sec:2:eq:AiI_0}
	\end{equation}
	
	In the sequel, we begin by analyzing the intersection of two spheres with a spherical shell, denoted by ${^{(3,2,1)}\mathcal{A}}_i = S^i_1 \cap S^i_2 \cap \mathcal{S}^i_3$, and subsequently characterize the resulting set $\mathcal{A}_i$ as defined in~\eqref{sec:2:eq:AiI_0}. In ${^{(u,v,w)}\mathcal{A}}_z$, the superscript $(u,v,w)$ denotes vertices $u, v, w \in U_z$ ordered according to the total order $\prec$ as $u \prec v \prec w \prec z$, where $z$ is the current vertex.
	
	\begin{figure}[!htp]
		\centering
		\begin{subfigure}[b]{0.49\linewidth}
			\centering
			\includegraphics[width=0.9\linewidth]{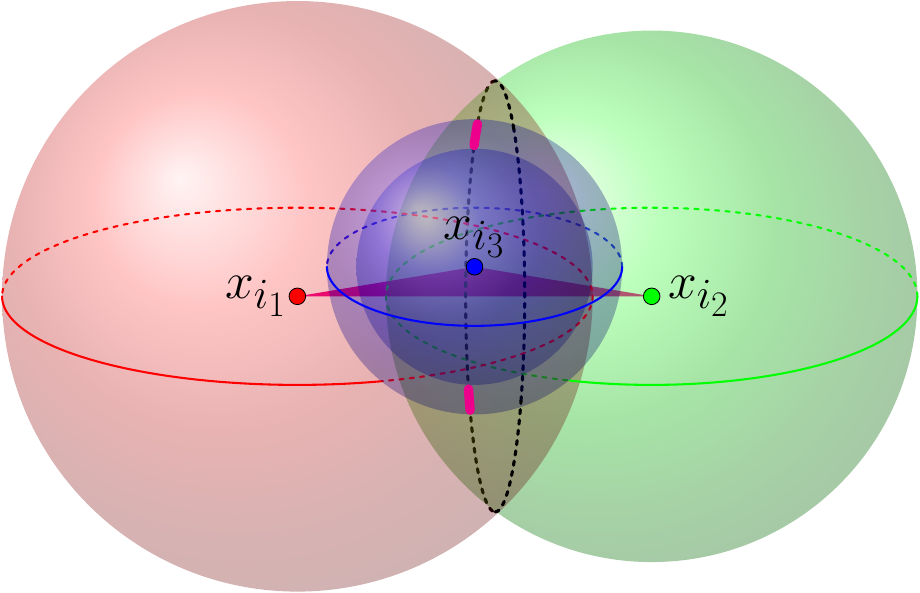}
			\caption{3D representation of $S_i^1 \cap S_i^2 \cap \mathcal{S}_i^3$}
		\end{subfigure}
		\begin{subfigure}[b]{0.49\linewidth}
			\centering
			\includegraphics[width=0.9\linewidth]{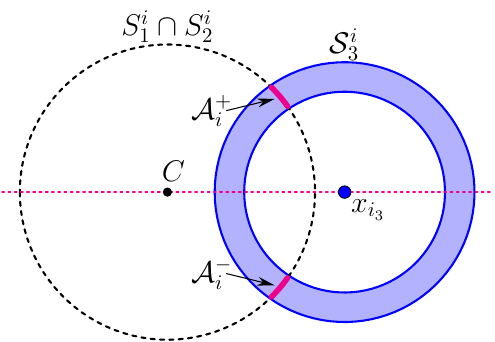}
			\caption{2D view of the plane containing $S_i^1 \cap S_i^2$}
		\end{subfigure}
		
		\caption{Illustration of the intersection ${^{(3,2,1)}\mathcal{A}}_i = S_i^1 \cap S_i^2 \cap \mathcal{S}_i^3$ and the resulting set $\mathcal{A}_i^- \cup \mathcal{A}_i^+$.}
		\label{fig:S1S2SS3}
	\end{figure}
	
	If ${^{(3,2,1)}\mathcal{A}}_i \neq \emptyset$, then this intersection lies on the circle defined by $S_i^1 \cap S_i^2$, as illustrated in~\autoref{fig:S1S2SS3}. It can be parameterized by considering the rotation of a point around the axis passing through ${x}_{i_2}$ and ${x}_{i_1}$. Once this parameterization is established, the final intersection is obtained by further intersecting this circle with the spherical shell $\mathcal{S}_i^3$. This can be accomplished by considering the torsion angle determined by the fixed points ${x}_{i_3},\ {x}_{i_2},\ {x}_{i_1}$ and a point ${x}_i \in {^{(3,2,1)}\mathcal{A}}_i$~\cite{2014_goncalves_aabsftbpatdg}. To compute this angle explicitly, the four points defining it must be specified. Since the only \textit{a priori} information available consists of the pairwise distances between these points, we recall that the cosine of the torsion angle defined by the four points $x_u,\ x_v,\ x_w$, and $x_z$ can be expressed solely in terms of these distances as follows~\cite{1987_pogorelov_g}:
	\begin{equation}
		\cos \tau_{u,v,w,z} = \dfrac{p(u,v,w,z) - d_{z,u}^2}{2q(u,v,w,z)},
		\label{sec:2:eq:costau}
	\end{equation}
	
	\noindent with
	\begin{align}
		\begin{aligned}
			p(u,v,w,z) & = \Big(\lambda(v,w,z) - \lambda(v,w,u)\Big)^2 + \rho^2(v,w,z) + \rho^2(v,w,u),\\[0.1cm]
			q(u,v,w,z) & = \rho(v,w,z)\rho(v,w,u),\\[0.1cm]
			\lambda(v,w,z) & = \dfrac{d_{z,v}^2 + d_{w,v}^2 - d_{z,w}^2}{2d_{w,v}}, \quad 
			\rho(v,w,z) = \sqrt{d_{z,v}^2 - \lambda^2(v,w,z)}.
		\end{aligned}
		\label{sec:2:eq:p_q_lambda_rho}
	\end{align}
	
	\begin{figure}[!htp]
		\centering
		\includegraphics[width=0.8\linewidth]{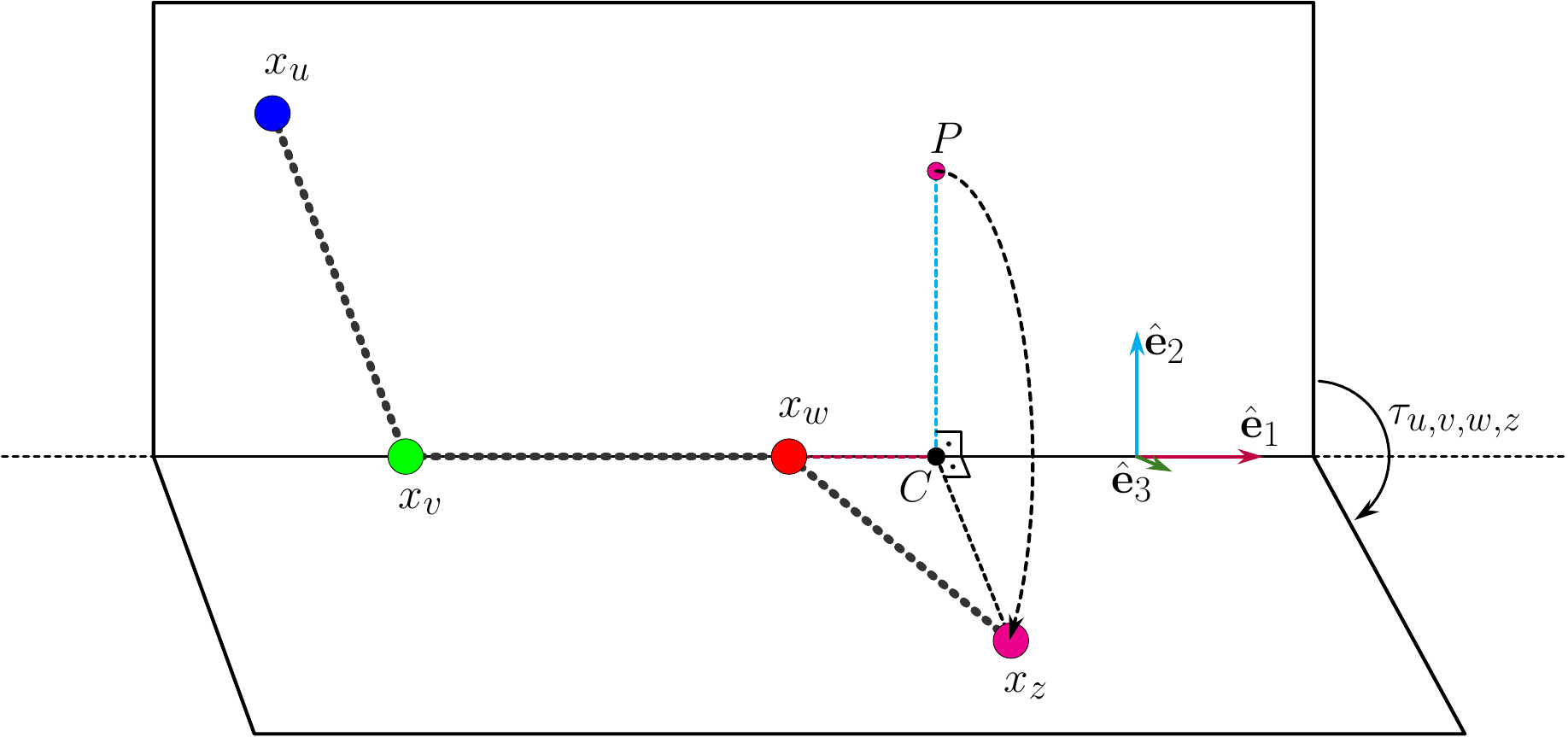}
		\caption{Geometric framework used to describe ${x}_z$ as a function of: the points ${x}_u,\ {x}_v,\ {x}_w$; the distances $d_{z,u},\ d_{z,v},\ d_{z,w}$; and the torsion angle $\tau_{u,v,w,z}$.}
		\label{fig:DDGP_framework}
	\end{figure}
	
	\autoref{fig:DDGP_framework} illustrates the framework for determining the position of $x_z$ based on the reference points $x_u,\ x_v,\ x_w$, the distances $d_{z,u},\ d_{z,v},\ d_{z,w}$, and the torsion angle $\tau_{u,v,w,z}$. The point $x_z$ results from the rotation of point $P$ about the axis $\hat{\mathbf{e}}_1$, through an angle $\tau_{u,v,w,z}$, with center at point $C$. Thus, by applying the fundamental principles of analytic geometry, the following expression for $x_z$ can be derived:
	\begin{equation}
		x_z(\tau_{u,v,w,z}) = \textbf{a}(v,w,z) + \textbf{b}(u,v,w,z) \cos\tau_{u,v,w,z} + \textbf{c}(u,v,w,z) \sin\tau_{u,v,w,z}
		\label{sec:2:eq:x_z},
	\end{equation}
	
	\noindent with
	\begin{equation}
		\textbf{a}(v,w,z) = x_v + \lambda(v,w,z)\hat{\mathbf{e}}_1,\ \textbf{b}(u,v,w,z) = \rho(v,w,z)\hat{\mathbf{e}}_2,\ \textbf{c}(u,v,w,z) = \rho(v,w,z)\hat{\mathbf{e}}_3,
		\label{sec:2:eq:a_b_c}\\[0.1cm]
	\end{equation}
	
	\noindent where $\hat{\mathbf{y}},\ \hat{\mathbf{e}}_1,\ \hat{\mathbf{e}}_2,\ \hat{\mathbf{e}}_3$ denote the normalized vectors associated with $\mathbf{y} = x_u - x_v$, $\mathbf{e}_1 = x_w - x_v$, $\mathbf{e_3} = \hat{\mathbf{e}}_1 \times \hat{\mathbf{y}}$, and $\mathbf{e}_2 = \hat{\mathbf{e}}_3 \times \hat{\mathbf{e}}_1$; $\lambda$ and $\rho$ are defined in~\eqref{sec:2:eq:p_q_lambda_rho}; $\cos \tau_{u,v,w,z}$ is given by~\eqref{sec:2:eq:costau}$;$ and $\sin \tau_{u,v,w,z} = \pm \sqrt{1 - \cos^2 \tau_{u,v,w,z}}$. Since the sine function is odd, the sign in $\sin\tau_{u,v,w,z}$ reflects the sign of the torsion angle $\tau_{u,v,w,z}$. Consequently, there are two possible configurations for $x_z$, corresponding to the two orientations of the torsion angle. As this sign cannot be determined \textit{a priori}, both possibilities must be considered. This idea has previously appeared in the literature, including within the framework of Clifford algebra~\cite{2014_goncalves_aabsftbpatdg, 2017_goncalvez_raotidgp, 2015_lavor_caatdmdgp, 2022_LGAAfSaSSIiR, 2026_CSaSSIwCGA}. Although the cross product is employed to construct the orthonormal basis $\{\hat{\mathbf{e}}_1,\hat{\mathbf{e}}_2,\hat{\mathbf{e}}_3\}$, thereby restricting the construction to $\mathbb{R}^3$, an equivalent orthonormal basis could be constructed using the Gram-Schmidt process~\cite{2016_strang_itla}, both in $\mathbb{R}^3$ and in higher-dimensional spaces. Nevertheless, the geometric construction considered here is inherently tied to three-dimensional space.
	
	The case in which $\mathcal{D}_{z,u}$ is given as an interval follows a similar treatment. In this setting, the torsion angle may span two branch intervals: one contained in $[-\pi, 0]$ and the other in $[0, \pi]$. The following lemma establishes a mapping from interval distances to corresponding angular intervals.
	\begin{lemma}
		\label{sec:2:lemma:1}
		Let $p,\ q \in \mathbb{R}_+$ such that $0 < 2q \leq p$. Also, let the functions $\tau^+ : \big[\sqrt{p - 2 q}, \sqrt{p + 2 q}\big] \rightarrow [0, \pi]$ and $\tau^- : \big[\sqrt{p - 2 q}, \sqrt{p + 2 q}\big] \rightarrow [-\pi,0]$ defined by
		\begin{equation*}
			\tau^\pm(d) = \pm\arccos\left(\dfrac{p - d^2}{2q}\right)
			\label{sec:2:eq:tau^+}.
		\end{equation*}
		
		\noindent Then $\tau^\pm$ are continuous and bijective. Moreover, $\tau^+$ is strictly increasing and $\tau^-$ is strictly decreasing. Consequently, and for an interval $\mathcal{D} = \big[\underline{d}, \overline{d}\big] \subset \big[\sqrt{p - 2 q}, \sqrt{p + 2 q}\big]$, we have
		$\tau^+(\mathcal{D}) = \big\{t \in [0, \pi] \ | \ \tau^+\big(\underline{d}\big) \leq t \leq \tau^+\big(\overline{d}\big)\big\} \equiv \mathcal{T}^+$ and
		$\tau^-(\mathcal{D}) = \big\{t \in [-\pi,0] \ | \ \tau^-\big(\overline{d}\big) \leq t \leq \tau^-\big(\underline{d}\big)\big\} \equiv \mathcal{T}^-$.
	\end{lemma}
	
	\begin{proof}
		Let $\alpha : \big[\sqrt{p - 2q}, \sqrt{p + 2q}\big] \rightarrow [-1,1]$ be defined by $\alpha(d) = \frac{p - d^2}{2q}$, and let $g : [-1,1] \rightarrow [0, \pi]$ be defined by $g(t) = \arccos(t)$. Note that $\alpha$ is continuous and strictly decreasing, while $g$ is continuous and strictly decreasing. Thus, $\tau^+(d) = (g \circ \alpha)(d)$ is continuous, bijective, and strictly increasing, whereas $\tau^-(d) = - (g \circ \alpha)(d)$ is continuous, bijective, and strictly decreasing. The expressions for $\tau^+(\mathcal{D})$ and $\tau^-(\mathcal{D})$ follow directly from continuity and monotonicity.
	\end{proof}
	
	The previous lemma establishes a well-defined and continuous relationship between the torsion angle $\tau_{u,v,w,z}$ and the corresponding distance $d_{z,u}$. From this point on, angular values are represented in $(-\pi,\pi]$. With this convention, when the boundary value occurs, the two branches give the representatives $-\pi$ and $\pi$ for the same angular position, and this value is kept only as $\pi$. The next proposition incorporates this convention into the distance-to-angle relation used for interval distances.
	\begin{proposition}
		\label{sec:2:proposition:1}
		Let $d_{v,u},\ d_{w,u},\ d_{w,v},\ d_{z,v},\ d_{z,w}$ be exact distances, and let $\mathcal{D}^0_{z,u}$ denote the interval distance originally associated with the pair $\{z,u\}$. Define the reachable distance interval
		\begin{equation*}
			\mathcal{I}_{u,v,w,z} = \big[\sqrt{p - 2 q}, \sqrt{p + 2 q}\big],
		\end{equation*}
		
		\noindent where $p = p(u,v,w,z)$ and $q = q(u,v,w,z)$ are given as in \eqref{sec:2:eq:p_q_lambda_rho}, and set
		\begin{equation*}
			\mathcal{D}_{z,u}
			= \mathcal{D}^0_{z,u} \cap \mathcal{I}_{u,v,w,z}.
		\end{equation*}
		
		\noindent If $\mathcal{D}_{z,u} = \emptyset$, then no point on the circle $S(x_v,d_{z,v})\cap S(x_w,d_{z,w})$ satisfies the interval distance constraint to $x_u$. Otherwise, write $\mathcal{D}_{z,u} = \big[\underline{d}_{z,u}, \overline{d}_{z,u}\big]$.
		Define the set-valued relation 
		\begin{equation*}
			\begin{array}{lcccl}
				R_{u,v,w,z} & : & \mathcal{I}_{u,v,w,z} & \rightrightarrows & (-\pi, \pi]\\[0.1cm]
				& & d & \longmapsto & R_{u,v,w,z}(d) = \left(\{\tau^-(d)\}\setminus\{-\pi\}\right)\cup\{\tau^+(d)\},
			\end{array}
		\end{equation*} 
		
		\noindent where $\tau^\pm$ are the functions introduced in the previous lemma. Then,
		\begin{equation*}
			R_{u,v,w,z}(\mathcal{D}_{z,u}) =
			\mathcal{T}^-_{u,v,w,z}
			\cup
			\mathcal{T}^+_{u,v,w,z},
			\label{sec:2:eq:tau_uvwz}
		\end{equation*}
		
		\noindent with $\mathcal{T}^-_{u,v,w,z} \equiv \Big[\tau^-\big(\overline{d}_{z,u}\big), \tau^-\big(\underline{d}_{z,u}\big)\Big]\setminus\{-\pi\}$ and $\mathcal{T}^+_{u,v,w,z} \equiv \Big[\tau^+\big(\underline{d}_{z,u}\big), \tau^+\big(\overline{d}_{z,u}\big)\Big]$.
	\end{proposition}
	
	\begin{proof}
		By~\eqref{sec:2:eq:p_q_lambda_rho}, we have
			\begin{equation*}
				0 < 2q(u,v,w,z)
				= 2\rho(v,w,z)\rho(v,w,u)
				\leq \rho^2(v,w,z)+\rho^2(v,w,u)
				\leq p(u,v,w,z),
			\end{equation*}
			
			\noindent where the middle inequality is immediate. Therefore, the result of the proposition follows directly from \Cref{sec:2:lemma:1}.
	\end{proof}
	
	The following lemma establishes how the point $x_z$ can be expressed as a function of the angular set defined in the preceding proposition.
	\begin{lemma}
		\label{sec:2:lemma:7}
		Let $x_u,\ x_v,\ x_w \in \mathbb{R}^3$ and $d_{v,u}$, $d_{w,u}$, $d_{w,v}$, $d_{z,v}$, $d_{z,w}$ be exact distances such that the two spheres centered at $x_v$ and $x_w$ intersect in a nondegenerate circle. Equivalently, $\rho(v,w,z)>0$, with $\rho$ defined in \eqref{sec:2:eq:p_q_lambda_rho}, and the vectors $\mathbf{b}$ and $\mathbf{c}$ in \eqref{sec:2:eq:a_b_c} are nonzero. Then, the function ${^{(u,v,w)}\chi_z} : (-\pi,\pi] \rightarrow \mathbb{R}^3$, defined by the expression in~\eqref{sec:2:eq:x_z}, is continuous and injective. Moreover, the image of any angular interval contained in this representative domain under ${^{(u,v,w)}\chi_z}$ is a circular arc lying in the plane spanned by the direction vectors $\mathbf{b}$ and $\mathbf{c}$, centered at the point $A$ whose position vector is $\mathbf{a}$ (given by $\mathbf{a} = A - O$, with $O$ denoting the origin of $\mathbb{R}^3$).
	\end{lemma}
	
	\begin{proof}
		The continuity of ${^{(u,v,w)}\chi_z}$ follows immediately from the expression in~\eqref{sec:2:eq:x_z}. By~\eqref{sec:2:eq:a_b_c}, the vectors $\mathbf{b}(u,v,w,z)$ and $\mathbf{c}(u,v,w,z)$ are orthogonal since $\hat{\textbf{e}}_2$ and $\hat{\textbf{e}}_3$ are orthonormal vectors, and because $\rho(v,w,z)>0$, they are nonzero and therefore linearly independent. Together with the fact that each $\theta \in (-\pi,\pi]$ is uniquely determined by its cosine and sine components, i.e., by the pair $(\cos\theta,\ \sin\theta)$, it follows that ${^{(u,v,w)}\chi_z}$ is injective. Therefore, ${^{(u,v,w)}\chi_z}$ maps each angle to a unique point in $\mathbb{R}^3$, and from~\eqref{sec:2:eq:x_z} its image over any angular interval contained in $(-\pi,\pi]$ is a circular arc centered at the point with position vector $\mathbf{a}$ and lying in the plane spanned by the direction vectors $\mathbf{b}$ and $\mathbf{c}$.
	\end{proof}
	
	The following classical result, presented as \Cref{sec:2:theorem:1}, is well known; detailed proofs can be found, for instance, in~\cite{1974_apostol_ma}.
	\begin{lemma}
		\label{sec:2:theorem:1}
		Let $f : A \rightarrow B$ be a function. If $A_\ell \subset A$ for all $\ell = 1,\ \dots,\ N$, then
		\begin{equation*}
			f\left(\bigcup_{\ell=1}^N A_\ell \right) = \bigcup_{\ell=1}^N f(A_\ell) \quad \text{and} \quad f\left(\bigcap_{\ell=1}^N A_\ell \right) \subseteq \bigcap_{\ell=1}^N f(A_\ell),
		\end{equation*}
				
		\noindent with equality if $f$ is injective.
	\end{lemma}
	
	With the above components in place, the intersection depicted in~\Cref{fig:S1S2SS3} can now be formally characterized.
	\begin{proposition}
		\label{sec:2:proposition:2}
		Let $S^w_z = S\big(x_w, d_{z,w}\big)$ and $S^v_z = S\big(x_v, d_{z,v}\big)$ be two spheres, and let $\mathcal{S}_z^u = S\big(x_u, \mathcal{D}^0_{z,u}\big)$ be a spherical shell. Let also $\mathcal{D}_{z,u}=\mathcal{D}^0_{z,u}\cap\mathcal{I}_{u,v,w,z}$ be the interval defined in \Cref{sec:2:proposition:1}. If $\mathcal{D}_{z,u}\neq\emptyset$, then
		\begin{equation*}
			S^w_z \cap S^v_z \cap \mathcal{S}^u_z = {^{(u,v,w)}\chi_z}\big(R_{u,v,w,z}(\mathcal{D}_{z,u})\big) = {^{(u,v,w)}\mathcal{A}}_z^- \cup {^{(u,v,w)}\mathcal{A}}_z^+,
			\label{sec:2:eq:proposition2}
		\end{equation*}
		
		\noindent where ${^{(u,v,w)}\chi_z}$ is the function defined in the preceding lemma, $R_{u,v,w,z}$ is the relation introduced in \Cref{sec:2:proposition:1}, and ${^{(u,v,w)}\mathcal{A}}_z^\pm$ are two arcs of the same circle, symmetric with respect to the plane $\{x_u,x_v,x_w\}$.
	\end{proposition}
	\begin{proof}
		As previously noted, the intersection $S^w_z \cap S^v_z \cap \mathcal{S}^u_z$ corresponds to the solution of the following nonlinear system in the variable $x_z$:
		\begin{equation}
			\left\{\begin{array}{ccccc}
				(\underline{d}^0_{z,u})^2 & \leq & \| x_z - x_u \|^2 & \leq & (\overline{d}^0_{z,u})^2\\[0.2cm]
				& & \| x_z - x_v \|^2 & = & d^2_{z,v}\\[0.2cm]
				& & \| x_z - x_w \|^2 & = & d^2_{z,w}
			\end{array}\right.
			\label{sec:2:eq:system_xzI_k=3}
		\end{equation}
		
		In this system, the precise $d_{z,v}$, $d_{z,w}$, and the interval $\mathcal{D}^0_{z,u} = \big[\underline{d}^0_{z,u}, \overline{d}^0_{z,u}\big]$ distances are given. Although the distances $d_{v,u}$, $d_{w,u}$, and $d_{w,v}$ do not appear explicitly in the system, they can be directly computed from the known positions of the points $x_u,\ x_v$, and $x_w$. Consequently, all parameters required to describe ${^{(u,v,w)}\chi}_z$ and $R_{u,v,w,z}$ are available, in accordance with~\eqref{sec:2:eq:p_q_lambda_rho}--\eqref{sec:2:eq:a_b_c}. If $\mathcal{D}_{z,u}=\emptyset$, then $S^w_z \cap S^v_z \cap \mathcal{S}^u_z = \emptyset$. Otherwise, it is straightforward to verify that any $x_z \in {^{(u,v,w)}\chi}_z\big(R_{u,v,w,z}(\mathcal{D}_{z,u})\big)$ satisfies the system~\eqref{sec:2:eq:system_xzI_k=3}. Therefore, \Cref{sec:2:proposition:1} and \Cref{sec:2:theorem:1} imply that
		\begin{align*}
			{^{(u,v,w)}\chi}_z\big(R_{u,v,w,z}(\mathcal{D}_{z,u})\big) & = {^{(u,v,w)}\chi}_z\big(\mathcal{T}_{u,v,w,z}^- \cup \mathcal{T}_{u,v,w,z}^+\big)\\[0.1cm]
			& = {^{(u,v,w)}\chi}_z\big(\mathcal{T}_{u,v,w,z}^-\big) \cup {^{(u,v,w)}\chi}_z\big(\mathcal{T}_{u,v,w,z}^+\big)\\[0.1cm]
			& = {^{(u,v,w)}\mathcal{A}}^-_z \cup {^{(u,v,w)}\mathcal{A}}^+_z,
			\label{sec:2:eq:x_z_A-UA+}
		\end{align*}
		
		\noindent where ${^{(u,v,w)}\mathcal{A}}^\pm_z$ denote two circular arcs, a result that follows directly from \Cref{sec:2:lemma:7}.
		
		Finally, since ${^{(u,v,w)}\chi}_z$ describes a circle lying in a plane orthogonal to the plane $\{x_u, x_v, x_w\}$ (see \autoref{fig:DDGP_framework}) with ${^{(u,v,w)}\chi}_z(0) = \mathbf{a}(v,w,z) + \mathbf{b}(u,v,w,z)$ lying in this plane (see equations \eqref{sec:2:eq:x_z} and \eqref{sec:2:eq:a_b_c}) and the angular intervals $\mathcal{T}_{u,v,w,z}^\pm$ being symmetric with respect to $0$, it follows that the union ${^{(u,v,w)}\mathcal{A}}^-_z \cup {^{(u,v,w)}\mathcal{A}}^+_z$ is invariant under reflection with respect to the plane $\{x_u, x_v, x_w\}$. Away from the endpoint represented by $\pi$, this reflection exchanges the two branch arcs.
	\end{proof}
	
	Although we have just described the solution of system \eqref{sec:2:eq:system_xiI_k>3} for the particular case where $\kappa_i = 3$, this result is sufficient to characterize the \textit{i}BP algorithm. This algorithm employs a sampling procedure that requires only the intersection of two spheres with a spherical shell. Therefore, in this section, we report the \textit{i}BP algorithm, while the general solution of system \eqref{sec:2:eq:system_xiI_k>3} will be deferred to the next section, as it forms the basis for the description of the newly proposed algorithm.
	
	To adapt the general notation of the previously established results to each vertex $i$ of the \textit{i}DDGP ordering, we identify $i_1 \leftarrow u$, $i_2 \leftarrow v$, $i_3 \leftarrow w$, and $i \leftarrow z$. Under this notation, and based on \Cref{sec:2:proposition:2}, the solution of $S_i^1 \cap S_i^2 \cap \mathcal{S}_i^3$ is given by
	\begin{equation*}
		S^1_i \cap S^2_i \cap \mathcal{S}^3_i = {^{(i_3,i_2,i_1)}\chi}_i\big(\mathcal{T}_i^0\big) = {^{(i_3,i_2,i_1)}\mathcal{A}}_i,
		\label{sec:2:eq:x_i123}
	\end{equation*}
	
	\noindent where
	\begin{equation}
		\mathcal{T}_i^0 = R_{i_3,i_2,i_1,i}(\mathcal{D}_{i,i_3}) = \mathcal{T}_{i_3,i_2,i_1,i}^- \cup \mathcal{T}_{i_3,i_2,i_1,i}^+,
		\label{sec:2:eq:T_i^0}
	\end{equation}
	
	\noindent with $R_{i_3,i_2,i_1,i}(\mathcal{D}_{i,i_3})$ defined in \Cref{sec:2:proposition:1}, where $\mathcal{D}_{i,i_3}$ denotes the intersection of the original interval distance with the corresponding reachable distance interval.
	
	To address system \eqref{sec:2:eq:system_xiI_k>3}, the \textit{i}BP algorithm constructs the symmetric arcs ${^{(i_3,i_2,i_1)}\mathcal{A}}_i$ using the equations derived from the exact distances $d_{i,i_1}$ and $d_{i,i_2}$, together with the inequality describing interval distance $\mathcal{D}_{i,i_3}$; if this reduced interval is empty, the current branch is pruned. It then samples these arcs using a fixed number of points. Each sampled point is subsequently checked for feasibility with respect to the remaining $\kappa_i - 3$ inequality/equality constraints in \eqref{sec:2:eq:system_xiI_k>3}. For each feasible point, the algorithm selects one at a time and proceeds to the next vertex in the \textit{i}DDGP order. If all sampled points are deemed infeasible, the algorithm backtracks to a previous vertex to select another feasible candidate before resuming the forward search. This process is repeated until the entire search space has been explored.
	
	To obtain a sample of the arc ${^{(i_3,i_2,i_1)}\mathcal{A}}_i$, we observe that it can be generated by sampling the interval $\mathcal{T}_i^0$, as defined in \eqref{sec:2:eq:T_i^0}, and applying the mapping ${^{(i_3,i_2,i_1)}\chi}_i$, introduced in \Cref{sec:2:lemma:7}, to each sampled value. That is, letting
	\begin{equation*}
		\begin{array}{lcl}
			T_i^0 & = & \text{sample}(\mathcal{T}_i^0) = \text{sample}(\mathcal{T}_{i_3,i_2,i_1,i}^- \cup \mathcal{T}_{i_3,i_2,i_1,i}^+)\\[0.3cm]
			& = &\text{sample}(\mathcal{T}_{i_3,i_2,i_1,i}^-) \cup \text{sample}(\mathcal{T}_{i_3,i_2,i_1,i}^+) \equiv T_i^- \cup T_i^+,
		\end{array}
		\label{sec:2:eq:Tipm}
	\end{equation*}
	
	\noindent we define the corresponding arc sample by
	\begin{equation*}
		{^{(i_3,i_2,i_1)}\mathbf{A}}_i = \text{sample}\big({^{(i_3,i_2,i_1)}\mathcal{A}}_i\big) = \big\{{^{(i_3,i_2,i_1)}\chi}_i(\tau)  \mid  \tau \in T_i^0 \big\}.
		\label{sec:2:eq:321Ai}
	\end{equation*}
	
	The strategy employed by the \textit{i}BP algorithm consists in attempting to construct a feasible subset $\mathbf{A}_i$ of \eqref{sec:2:eq:AiI_0} according to the following procedure:
	\begin{equation*}
		\mathbf{A}_i \subset \mathcal{A}_i, \text{ where } \mathbf{A}_i = {^{(i_3,i_2,i_1)}\mathbf{A}}_i \cap \mathcal{S}_i^4 \cap \dots \cap \mathcal{S}_i^{\kappa_i}.
		\label{sec:2:eq:Ai_iBP}
	\end{equation*}
	
	\Cref{sec:2:fig:Ai_iBP} illustrates an example in which $\mathbf{A}_i \neq \emptyset$, for a case where $\kappa_i = 4$ and $|T_i^\pm| = 5$, i.e., $|T_i^0| = 10$.
	\begin{figure}[!htp]
		\centering
		\begin{subfigure}[b]{0.49\linewidth}
			\centering
			\includegraphics[width=0.9\linewidth]{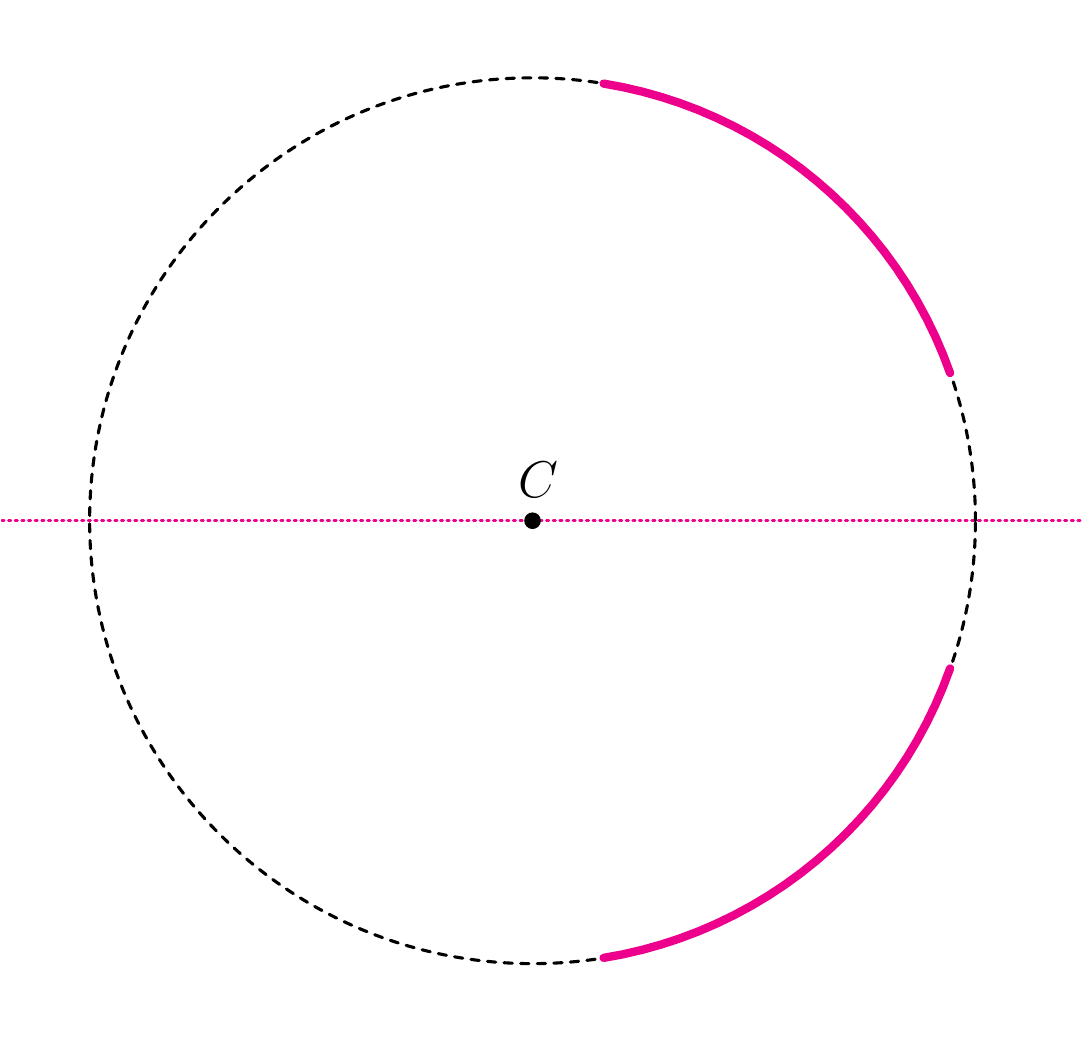}
			\caption{${^{(i_3,i_2,i_1)}\mathcal{A}}_i = S_i^1 \cap S_i^2 \cap \mathcal{S}_i^3
				$}
		\end{subfigure}
		\begin{subfigure}[b]{0.49\linewidth}
			\centering
			\includegraphics[width=0.9\linewidth]{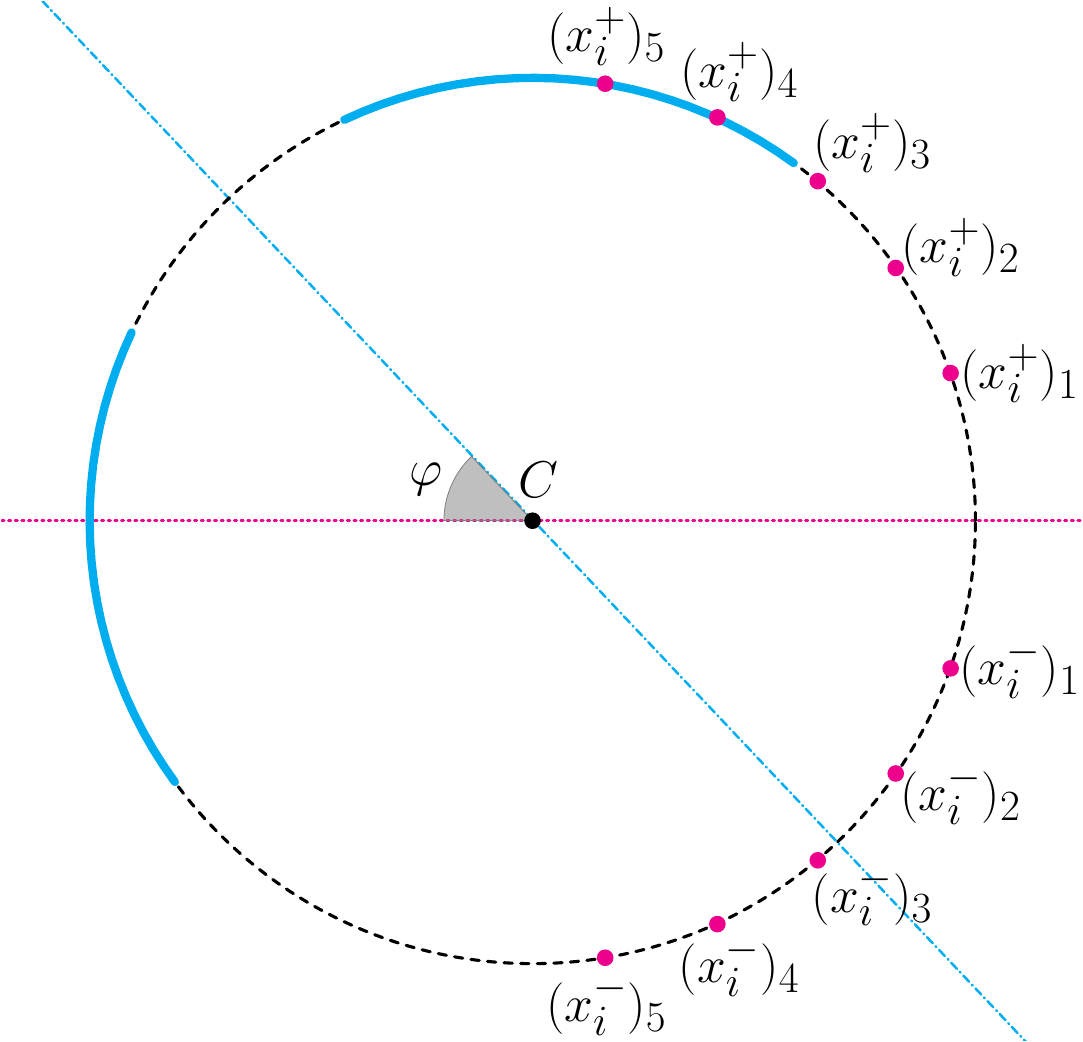}
			\caption{$\mathbf{A}_i = {^{(i_3,i_2,i_1)}\mathbf{A}}_i \cap \mathcal{S}_i^4$}
		\end{subfigure}
		
		\caption{View of the plane containing the intersection $S_i^1 \cap S_i^2$, represented as a dashed black circumference, where the angle $\varphi$ denotes the torsion between the planes $\{x_{i_4}, x_{i_2}, x_{i_1}\}$ and $\{x_{i_3}, x_{i_2}, x_{i_1}\}$. In (a), the magenta arcs correspond to the intersection of $\mathcal{S}_i^3$ with the fixed circumference, while in (b) the cyan arcs correspond to the intersection of $\mathcal{S}_i^4$ with the same circumference. A sample of $S_i^1 \cap S_i^2 \cap \mathcal{S}_i^3$ is represented by the set ${^{(i_3,i_2,i_1)}\mathbf{A}}_i = \big\{(x_i^-\big)_1,\ \dots,\ \big(x_i^-\big)_5,\ \big(x_i^+\big)_1,\ \dots,\ \big(x_i^+\big)_5\big\}$. Its feasible subset is $\mathbf{A}_i = {^{(i_3,i_2,i_1)}\mathbf{A}}_i \cap \mathcal{S}_i^4 = \big\{\big(x_i^+\big)_4, \ \big(x_i^+\big)_5\big\}$.}
		\label{sec:2:fig:Ai_iBP}
	\end{figure}
	
	As illustrated in \Cref{sec:2:fig:Ai_iBP}, the strategy employed by the \textit{i}BP algorithm to identify points that satisfy \eqref{sec:2:eq:AiI_0} may yield a significant number of infeasible candidates, potentially causing the algorithm to return without a solution. It is also worth noting that the outcome of the algorithm is highly sensitive to the sampling strategy adopted. In the example shown, the sampling over $\mathcal{T}_i^0$ was performed uniformly with equally spaced points. On the positive side, the feasibility of each sampled point can be verified efficiently. A pseudocode summarizing the procedure discussed in this section is provided in \Cref{sec:2:pc:iBP}, and computational results illustrating the performance of the algorithm are presented in \Cref{sec:5}.

\begin{algorithm}[!htp]
	\caption{A main framework of the \textit{i}BP algorithm in $\mathbb{R}^3$.}
	\label{sec:2:pc:iBP}
	\begin{algorithmic}[1]\small
		\Require{$G=(V,E,\mathcal{D})$, $X=\{x_1,x_2,x_3\}$, predecessor sets $U=[U_4,\ \dots,\ U_{|V|}]$ with $U_i=\{i_1,i_2,i_3,\dots,i_{\kappa_i}\}$, and sampling sizes $N=[\nu_4,\ \dots,\ \nu_{|V|}]$ with $\nu_i\geq 1$}
		\Ensure{A feasible sampled realization $X$, or $\emptyset$ if no feasible realization is found in the finite sampled search}

		\Function{$\mathit{i}\mathrm{BP}$}{$G,X,U,N$}
		\State $i\leftarrow 4$;
		\For{$j=4,\dots,|V|$}
			\State $T_j\leftarrow\emptyset$ and $r_j\leftarrow0$;
		\EndFor
		\While{$i\geq 4$}
			\If{$i=|V|+1$}
				\State \Return $X$;
			\ElsIf{$r_i=0$}
				\State Compute the reduced interval $\mathcal{D}_{i,i_3}$ and the angular set $\mathcal{T}_i^0$ as in \eqref{sec:2:eq:T_i^0};
				\If{$\mathcal{T}_i^0=\emptyset$}
					\State $T_i\leftarrow\emptyset$;
				\Else
					\State Compute a finite sample $T_i\subset\mathcal{T}_i^0$ satisfying $|T_i|\leq \nu_i$;
					\State Set $\mathbf{a}_i\leftarrow\mathbf{a}(i_2,i_1,i)$, $\mathbf{b}_i\leftarrow\mathbf{b}(i_3,i_2,i_1,i)$, and $\mathbf{c}_i\leftarrow\mathbf{c}(i_3,i_2,i_1,i)$ as in \eqref{sec:2:eq:a_b_c};
				\EndIf
				\State $r_i\leftarrow1$;
			\ElsIf{$T_i=\emptyset$ or $r_i>|T_i|$}
				\State $T_i\leftarrow\emptyset$ and $r_i\leftarrow0$;
				\State $i\leftarrow i-1$;
				\If{$i\geq4$}
					\State $X\leftarrow X\setminus\{x_i\}$ and $r_i\leftarrow r_i+1$;
				\EndIf
			\Else
				\State $\tau\leftarrow$ the $r_i$-th torsion angle in $T_i$;
				\State $x_i\leftarrow \mathbf{a}_i+\mathbf{b}_i\cos\tau+\mathbf{c}_i\sin\tau$;
				\If{$X\cup\{x_i\}$ satisfies the distance constraints whose endpoints are already embedded}
					\State $X\leftarrow X\cup\{x_i\}$;
					\State $i\leftarrow i+1$;
					\If{$i\leq |V|$}
						\State $T_i\leftarrow\emptyset$ and $r_i\leftarrow0$;
					\EndIf
				\Else
					\State $r_i\leftarrow r_i+1$;
				\EndIf
			\EndIf
		\EndWhile
		\State \Return $\emptyset$;
		\EndFunction
	\end{algorithmic}
\end{algorithm}
	
	\section{A Torsion-Angle-Based Approach for the \textit{i}DDGP}\label{sec:3}
	
	The central idea of this new algorithm is to fully describe the set $\mathcal{A}_i$ defined in \eqref{sec:2:eq:AiI_0}, and subsequently sample from it, rather than attempting to construct a sample based on the intersection $S_i^1 \cap S_i^2 \cap \mathcal{S}_i^3$. Once this description is obtained, the algorithm proceeds analogously to the \textit{i}BP strategy: a sampled point is selected, and the exploration continues along the corresponding branch of the search tree. To this end, we rely on the following insightful reformulation of \eqref{sec:2:eq:AiI_0}, which serves as a key step in enabling this approach:
	\begin{equation*}
		\mathcal{A}_i =  S_i^1 \cap S_i^2 \cap \mathcal{S}_i^3 \cap \dots \cap \mathcal{S}_i^{\kappa_i}
		=
		\big(S_i^1 \cap S_i^2 \cap \mathcal{S}_i^3\big) \cap \dots \cap \big(S_i^1 \cap S_i^2 \cap \mathcal{S}_i^{\kappa_i}\big)
		=
		\displaystyle \bigcap_{\ell=3}^{\kappa_i} \big(S_i^1 \cap S_i^2 \cap \mathcal{S}_i^\ell\big).
		\label{sec:2:eq:AiI_1}
	\end{equation*}
	
	\noindent By applying \Cref{sec:2:proposition:2} to each intersection $S_i^1 \cap S_i^2 \cap \mathcal{S}_i^\ell$ in the previous expression, we obtain
	\begin{equation*}
		\mathcal{A}_i = \bigcap_{\ell=3}^{\kappa_i} \big(S^1_i \cap S^2_i \cap \mathcal{S}^\ell_i\big) = \bigcap_{\ell=3}^{\kappa_i} \left({^{(i_\ell,i_2,i_1)}\chi}_i\big(\mathcal{T}_{i_\ell,i_2,i_1,i}\big)\right) = \bigcap_{\ell=3}^{\kappa_i} \left({^{(i_\ell,i_2,i_1)}\mathcal{A}}_i^- \cup {^{(i_\ell,i_2,i_1)}\mathcal{A}}_i^+\right)
		\label{sec:3:eq:AiI_2}
	\end{equation*}
	
	As previously discussed, $\mathcal{A}_i$ is a solution to \eqref{sec:2:eq:system_xiI_k>3}. Therefore, we conclude that the solution set of \eqref{sec:2:eq:system_xiI_k>3} corresponds to the intersection of $\kappa_i - 2$ unions of circular arcs. In what follows, we elucidate the procedure for computing this intersection.
	
	The parametrization above uses the representative domain $(-\pi,\pi]$. Angular translations, however, are circular operations and must be performed modulo $2\pi$ before being written again in this same domain. Thus, for $a,b\in(-\pi,\pi]$, define
	\begin{equation*}
		\oplus : (-\pi,\pi]\times(-\pi,\pi]\to(-\pi,\pi]
	\end{equation*}
		
	\noindent by
	\begin{equation*}
		a\oplus b := (a+b)\pmod{2\pi},
		\qquad \text{represented in }(-\pi,\pi].
	\end{equation*}
		
	This value is well defined: every real angle can be translated by an integer multiple of $2\pi$ into $(-\pi,\pi]$, and two representatives in $(-\pi,\pi]$ that are congruent modulo $2\pi$ must coincide. With this operation in place, the following proposition can be stated:
	\begin{proposition}
		\label{sec:3:proposition:3}
		Let $d_{w,u_0}$, $d_{v,u_0}$, $d_{w,u_1}$, $d_{v,u_1}$, $d_{w,v}$, $d_{z,v}$, $d_{z,w}$ be exact distances, and let ${^{(u,v,w)}\chi}_z$ be the function defined in \Cref{sec:2:lemma:7}. Then,
		\begin{equation*}
			{^{(u_1,v,w)}\chi}_z(\tau) = {^{(u_0,v,w)}\chi}_z\big(\varphi_{u_0,v,w,u_1} \oplus \tau\big),
		\end{equation*}
		
		\noindent where $\varphi_{u_0,v,w,u_1}$ is the torsion angle between the planes $\{x_{u_0}, x_v, x_w\}$ and $\{x_{u_1}, x_v, x_w\}$.
	\end{proposition}
	
	\begin{proof}
		Since we are dealing with two variations of the function ${^{(u,v,w)}\chi}_z$, which share many common variables, we introduce the following simplified notation:
		\begin{equation*}
			\varphi_{0,1} := \varphi_{u_0,v,w,u_1},\ \mathbf{a} := \mathbf{a}(v,w,z),\ \mathbf{b}_\ell := \mathbf{b}(u_\ell,v,w,z),\ \mathbf{c}_\ell := \mathbf{c}(u_\ell,v,w,z),\ \text{for } \ell = 0, 1.
			\label{sec:3:eq:proof_proposition3_1}
		\end{equation*}
		
		First, observe that, by the definition of $\oplus$, the value $\theta_1\oplus\theta_2$ is the representative in $(-\pi,\pi]$ of an angle that differs from $\theta_1+\theta_2$ by an integer multiple of $2\pi$. Therefore, since the sine and cosine functions are $2\pi$-periodic, the trigonometric values used in ${^{(u_0,v,w)}\chi}_z$ are the same as those obtained from the ordinary sum $\theta_1+\theta_2$. Hence,
		\begin{align*}
			{^{(u_0,v,w)}\chi}_z\big(\varphi_{0,1} \oplus \tau\big) & = \mathbf{a} + \mathbf{b}_0 \cos (\varphi_{0,1} \oplus \tau) + \mathbf{c}_0 \sin (\varphi_{0,1} \oplus \tau) \\[0.1cm]
			& = \mathbf{a} + \mathbf{b}_0 \cos (\varphi_{0,1} + \tau) + \mathbf{c}_0 \sin (\varphi_{0,1} + \tau) \\[0.1cm]
			& = \mathbf{a} + \big(\mathbf{b}_0 \cos\varphi_{0,1} + \mathbf{c}_0 \sin\varphi_{0,1}\big)\cos\tau + \big(\mathbf{c}_0 \cos\varphi_{0,1} - \mathbf{b}_0 \sin\varphi_{0,1}\big)\sin\tau.
		\end{align*}
		
		As ${^{(u_1,v,w)}\chi}_z(\tau) = \mathbf{a} + \mathbf{b}_1 \cos\tau + \mathbf{c}_1 \sin\tau$, the proposition follows if we verify that:
		\begin{equation}
			\left\{\begin{array}{lcl}
				\mathbf{b}_1 & = & \mathbf{b}_0 \cos \varphi_{0,1} + \mathbf{c}_0 \sin \varphi_{0,1}, \\[0.2cm]
				\mathbf{c}_1 & = & \mathbf{c}_0 \cos \varphi_{0,1} - \mathbf{b}_0 \sin \varphi_{0,1}.
			\end{array}\right.
			\label{sec:3:eq:proof_proposition3_4}
		\end{equation}
		
		From the expressions of $\mathbf{b}(u,v,w,z)$ and $\mathbf{c}(u,v,w,z)$ in \eqref{sec:2:eq:a_b_c}, we note the following identities:
		\begin{align}
			& \|\mathbf{b}_\ell\| = \|\mathbf{c}_\ell\| = \rho \tag{P1}\label{P1} \\[0.1cm]
			& \mathbf{c}_\ell = \hat{\mathbf{e}}_1 \times \mathbf{b}_\ell \tag{P2}\label{P2}\\[0.1cm]
			& \hat{\mathbf{e}}_1 \perp \mathbf{b}_\ell \tag{P3}\label{P3}
		\end{align}
		
		\noindent for $\ell = 0, 1$.
		
		Since $\varphi_{0,1}$ represents the angle between the planes $\{x_{u_0}, x_v, x_w\}$ and $\{x_{u_1}, x_v, x_w\}$, and $\textbf{c}_0$ and $\textbf{c}_1$ are the normal vectors to these planes, it follows that $\varphi_{0,1}$ is the angle between the vectors $\textbf{c}_0$ and $\textbf{c}_1$. To express $\varphi_{0,1}$ in terms of $\textbf{b}_\ell$ and $\textbf{c}_\ell$, we apply the standard vector identities in $\mathbb{R}^3$ as outlined in \cite{2004_thornton_cdopas, 2011_zill_aem}.
		\begin{align}
			& (\mathbf{n}_1 \times \mathbf{n}_2)\times(\mathbf{n}_3 \times \mathbf{n}_4) = [(\mathbf{n}_1 \times \mathbf{n}_2)\cdot \mathbf{n}_4]\mathbf{n}_3 - [(\mathbf{n}_1 \times \mathbf{n}_2) \cdot \mathbf{n}_3]\mathbf{n}_4 \tag{P4}\label{P4}\\[0.1cm]
			& \mathbf{n}_1 \cdot (\mathbf{n}_2 \times \mathbf{n}_3) = \mathbf{n}_3 \cdot (\mathbf{n}_1 \times \mathbf{n}_2) \tag{P5}\label{P5}\\[0.1cm]
			& \mathbf{n}_1 \times \mathbf{n}_2 = -\mathbf{n}_2 \times \mathbf{n}_1 \tag{P6}\label{P6}\\[0.1cm]
			& \|\mathbf{n}_1\|^2 = \|\mathbf{n}_2\|^2 = \mathbf{n}_1 \cdot \mathbf{n}_2 \ \Leftrightarrow\ \mathbf{n}_1 = \mathbf{n}_2 \tag{P7}\label{P7}
		\end{align}
		
		Thus, based on the definitions of the dot and cross products, for representatives of $\varphi_{0,1}$ on the real branches $[-\pi,0]$ and $[0,\pi]$, we obtain:
		\begin{equation*}
			\begin{array}{rclcl}
				\rho^2 \cos \varphi_{0,1}
				& \overset{\eqref{P1}}{=} & \|\mathbf{c}_0\|\|\mathbf{c}_1\| \cos \varphi_{0,1} = \mathbf{c}_0 \cdot \mathbf{c}_1\\[0.2cm]
				& \overset{\eqref{P1}}{=} & \|\mathbf{b}_0\|\|\mathbf{b}_1\| \cos \varphi_{0,1} = \mathbf{b}_0 \cdot \mathbf{b}_1\\[0.2cm]
				\rho^2 \sin |\varphi_{0,1}| &
				\overset{\eqref{P1}}{=} & \|\mathbf{c}_0\|\|\mathbf{c}_1\|\sin |\varphi_{0,1}| = \|\mathbf{c}_0 \times \mathbf{c}_1\| \overset{\eqref{P2}}{=} \|(\hat{\mathbf{e}}_1 \times \mathbf{b}_0) \times (\hat{\mathbf{e}}_1 \times \mathbf{b}_1)\|\\[0.2cm]
				& \overset{\eqref{P4}}{=} &\|[(\hat{\mathbf{e}}_1 \times \mathbf{b}_0) \cdot \mathbf{b}_1] \hat{\mathbf{e}}_1 - [(\hat{\mathbf{e}}_1 \times \mathbf{b}_0) \cdot \hat{\mathbf{e}}_1] \mathbf{b}_1\| \overset{\eqref{P3}}{=} |(\hat{\mathbf{e}}_1 \times \mathbf{b}_0) \cdot \mathbf{b}_1|\|\hat{\mathbf{e}}_1\|\\[0.2cm]
				& \overset{\eqref{P2}}{=} & |\mathbf{c}_0 \cdot \mathbf{b}_1|
			\end{array}
		\end{equation*}
		
		By definition of $\mathbf{c}_0$ and $\mathbf{b}_1$, the sign of $\mathbf{c}_0 \cdot \mathbf{b}_1$ encodes the orientation of the torsion angle $\varphi_{0,1}$ on the two real branches: $\mathbf{c}_0 \cdot \mathbf{b}_1 > 0$ if and only if $\varphi_{0,1} \in (0,\pi)$, whereas $\mathbf{c}_0 \cdot \mathbf{b}_1 < 0$ if and only if $\varphi_{0,1} \in (-\pi,0)$. The remaining case $\mathbf{c}_0 \cdot \mathbf{b}_1=0$ corresponds to a coplanar boundary value, namely $\varphi_{0,1}=0$ or the endpoint orientation represented by $\pi$ in $(-\pi,\pi]$. This geometric relation can be observed in \autoref{fig:DDGP_framework} by identifying $u \leftarrow u_0$, $v \leftarrow v$, $w \leftarrow w$, and $z \leftarrow u_1$. Accordingly, for these branch representatives of $\varphi_{0,1}$, we can write
		\begin{equation*}
			\begin{array}{lcl}
				\rho^2 \sin |\varphi_{0,1}| = |\mathbf{c}_0 \cdot \mathbf{b}_1| \Leftrightarrow \rho^2 \sin \varphi_{0,1} & = & \mathbf{c}_0 \cdot \mathbf{b}_1\\[0.2cm]
				& \overset{\eqref{P2}}{=} & (\hat{\mathbf{e}}_1 \times \mathbf{b}_0) \cdot \mathbf{b}_1 \overset{\eqref{P5}}{=} (\mathbf{b}_1 \times \hat{\mathbf{e}}_1) \cdot \mathbf{b}_0\\[0.2cm]
				& \overset{\eqref{P6}}{=} & -(\hat{\mathbf{e}}_1 \times \mathbf{b}_1) \cdot \mathbf{b}_0 \overset{\eqref{P2}}{=} -\mathbf{c}_1 \cdot \mathbf{b}_0
			\end{array}
		\end{equation*}
		
		Therefore, we conclude that $\mathbf{c}_0 \cdot \mathbf{c}_1 = \mathbf{b}_0 \cdot \mathbf{b}_1 = \rho^2 \cos \varphi_{0,1}$ and $\mathbf{c}_0 \cdot \mathbf{b}_1 = - \mathbf{c}_1 \cdot \mathbf{b}_0 = \rho^2 \sin \varphi_{0,1}$, where the trigonometric values are independent of the chosen representative of $\varphi_{0,1}$. Consequently,
		\begin{equation*}
			\begin{array}{rcl}
				\mathbf{b}_1 \cdot \big(\mathbf{b}_0 \cos \varphi_{0,1} + \mathbf{c}_0 \sin \varphi_{0,1}\big) & = & \mathbf{b}_1 \cdot \mathbf{b}_0 \cos \varphi_{0,1} + \mathbf{b}_1 \cdot \mathbf{c}_0 \sin \varphi_{0,1}\\[0.2cm]
				& = & \rho^2 \cos^2 \varphi_{0,1} + \rho^2 \sin^2 \varphi_{0,1} = \rho^2,\\[0.3cm]
				\mathbf{c}_1 \cdot \big(\mathbf{c}_0 \cos \varphi_{0,1} - \mathbf{b}_0 \sin \varphi_{0,1}\big) & = & \mathbf{c}_1 \cdot \mathbf{c}_0 \cos \varphi_{0,1} - \mathbf{c}_1 \cdot \mathbf{b}_0 \sin \varphi_{0,1}\\[0.2cm]
				& = & \rho^2 \cos^2 \varphi_{0,1} + \rho^2 \sin^2 \varphi_{0,1} = \rho^2.
			\end{array}
		\end{equation*}

		By \eqref{P1} and \eqref{P2}, $\mathbf{b}_0\perp\mathbf{c}_0$ and
		\begin{equation*}
			\|\mathbf{b}_0\|=\|\mathbf{c}_0\|=\|\mathbf{b}_1\|=\|\mathbf{c}_1\|=\rho .
		\end{equation*}
		
		\noindent Therefore, the vectors
		$\mathbf{b}_0\cos\varphi_{0,1}+\mathbf{c}_0\sin\varphi_{0,1}$ and $\mathbf{c}_0\cos\varphi_{0,1}-\mathbf{b}_0\sin\varphi_{0,1}$ also have norm $\rho$. Together with the identities obtained above, this gives
		\begin{equation}
			\begin{array}{lclcl}
				\|\mathbf{b}_1\|^2
				& = &
				\|\mathbf{b}_0\cos\varphi_{0,1}+\mathbf{c}_0\sin\varphi_{0,1}\|^2
				& = &
				\mathbf{b}_1 \cdot \big(\mathbf{b}_0 \cos \varphi_{0,1} + \mathbf{c}_0 \sin \varphi_{0,1}\big),\\[0.2cm]
				\|\mathbf{c}_1\|^2
				& = & 
				\|\mathbf{c}_0\cos\varphi_{0,1}-\mathbf{b}_0\sin\varphi_{0,1}\|^2
				& = &
				\mathbf{c}_1 \cdot \big(\mathbf{c}_0 \cos \varphi_{0,1} - \mathbf{b}_0 \sin \varphi_{0,1}\big).
			\end{array}
			\label{sec:3:eq:proof_proposition3_5}
		\end{equation}
		\noindent Applying \eqref{P7} to the two identities in \eqref{sec:3:eq:proof_proposition3_5} yields the equations in \eqref{sec:3:eq:proof_proposition3_4}, and the proposition is proved.
	\end{proof}
	
	From the previous proposition, we may directly derive the following corollary:
	\begin{corollary}
		\label{sec:3:corollary:1}
		Let ${^{(u,v,w)}\chi}_z$, $\oplus$, and $\varphi_{u_0,v,w,u_1}$ denote, respectively, the function, the angular addition modulo $2\pi$ written in $(-\pi,\pi]$, and the torsion angle as employed in the previous proposition. If $\mathcal{T}=[\underline{\tau},\overline{\tau}]\subset(-\pi,\pi]$, then
		\begin{equation*}
			{^{(u_1,v,w)}\chi}_z\big(\mathcal{T}\big) = {^{(u_0,v,w)}\chi}_z\big(\varphi_{u_0,v,w,u_1} \oplus \mathcal{T}\big),
			\label{sec:3:eq:corollary:1}
		\end{equation*}
		
		\noindent where $\varphi \oplus \mathcal{T}$ denotes the circular translation of $\mathcal{T}$ by $\varphi$, written with representatives in $(-\pi,\pi]$,
		\begin{equation*}
			\varphi \oplus \mathcal{T}
			:=
			\big\{\varphi\oplus\tau \mid \tau\in\mathcal{T}\big\}\subset(-\pi,\pi].
		\end{equation*}

		\noindent Let $\underline{\omega}=\varphi\oplus\underline{\tau}$ and $\overline{\omega}=\varphi\oplus\overline{\tau}$. When represented in $(-\pi,\pi]$, this set is given by
		\begin{equation*}
			\varphi\oplus\mathcal{T}
			=
			\left\{\begin{array}{lcl}
				[\underline{\omega},\overline{\omega}], & \text{if} & \underline{\omega} \leq \overline{\omega},\\[0.2cm]
				(-\pi,\overline{\omega}]\cup[\underline{\omega},\pi], & \text{if} & \underline{\omega} > \overline{\omega}.
			\end{array}\right.
		\end{equation*}

		\noindent Thus, a circularly translated interval is represented either as one ordinary interval or as the union of two ordinary intervals.		
	\end{corollary}
	
	Taking into account all the results developed in this section, we now state the main result of this section, which characterizes the solution to the nonlinear system \eqref{sec:2:eq:system_xiI_k>3}:
	
	\begin{proposition}
		\label{sec:3:proposition:4}
		The solution to the nonlinear system \eqref{sec:2:eq:system_xiI_k>3}, whose equivalent geometric formulation involves intersecting the spheres \(S_1^i = S(x_{i_1}, d_{i,i_1})\), \(S_2^i = S(x_{i_2}, d_{i,i_2})\), and the spherical shells $\mathcal{S}_3^i = S\left(x_{i_3}, \mathcal{D}^0_{i,i_3}\right),\ \dots,\ \mathcal{S}_{\kappa_i}^i = S\left(x_{i_{\kappa_i}}, \mathcal{D}^0_{i,i_{\kappa_i}}\right)$, is given by
		\begin{equation}
			\mathcal{A}_i = {^{\left(i_3,i_2,i_1\right)}\chi}_i\big(\mathcal{T}_i\big), \text{ such that }
			\mathcal{T}_i = \mathcal{T}_i^0 \cap \left(\bigcap_{\ell = 4}^{\kappa_i} \varphi_{i_3,i_2,i_1,i_\ell} \oplus \mathcal{T}_{i_\ell,i_2,i_1,i}\right)
			\label{sec:3:eq:proposition:4:T_i},
		\end{equation}
		
		\noindent where
		\begin{itemize}
			\item ${^{(u,v,w)}}\chi_z$ denotes the mapping introduced in \Cref{sec:2:lemma:7};
			\item $\mathcal{T}_i^0$ refers to the angular set defined in \eqref{sec:2:eq:T_i^0};
			\item $\mathcal{T}_{u,v,w,z} := \mathcal{T}_{u,v,w,z}^- \cup \mathcal{T}_{u,v,w,z}^+$ is the angular set characterized in \Cref{sec:2:proposition:1};
			\item $\varphi_{u,v,w,z}$ denotes the torsion angle between the planes $\{x_u, x_v, x_w\}$ and $\{x_v, x_w, x_z\}$.
		\end{itemize}
	\end{proposition}
	
	\begin{proof}
		We have
		\begin{equation*}
			\begin{array}{lcl}
				\mathcal{A}_i & = & \displaystyle \bigcap_{\ell = 3}^{\kappa_i} \left( S_i^1 \cap S_i^2 \cap \mathcal{S}_i^\ell \right)
				\overset{\text{Prop. \ref{sec:2:proposition:2}}}{=} \bigcap_{\ell = 3}^{\kappa_i} {^{(i_\ell,i_2,i_1)}\chi}_i\left( \mathcal{T}_{i_\ell,i_2,i_1,i} \right)\\[0.4cm]
				& \overset{\text{Cor. \ref{sec:3:corollary:1}}}{=} & \displaystyle \left({^{(i_3,i_2,i_1)}\chi}_i\left(\mathcal{T}_{i_3,i_2,i_1,i} \right)\right) \cap \left(\bigcap_{\ell = 4}^{\kappa_i} {^{(i_3,i_2,i_1)}\chi}_i\left( \varphi_{i_3,i_2,i_1,i_\ell} \oplus \mathcal{T}_{i_\ell,i_2,i_1,i} \right)\right) \\[0.4cm]
				& \overset{{\text{Lem. \ref{sec:2:theorem:1}}}}{=} & \displaystyle {^{(i_3,i_2,i_1)}\chi}_i \left(\mathcal{T}_i^0 \cap \left(\bigcap_{\ell = 4}^{\kappa_i} \varphi_{i_3,i_2,i_1,i_\ell} \oplus \mathcal{T}_{i_\ell,i_2,i_1,i}\right)\right).
			\end{array}
			\label{sec:3:eq:proof_proposition4_1}
		\end{equation*}
		
		\noindent Hence, the solution is given by $\mathcal{A}_i = {^{(i_3,i_2,i_1)}\chi}_i(\mathcal{T}_i)$, where $\mathcal{T}_i$ is defined as in \eqref{sec:3:eq:proposition:4:T_i}.
	\end{proof}
	
	In \eqref{sec:3:eq:proposition:4:T_i}, if an angular set is represented as a finite union of interval components, each component is first represented after the corresponding $\oplus$ operation in $(-\pi,\pi]$, and the resulting finite unions are then used in the intersections defining $\mathcal{T}_i$.
	
	Regarding the computation of $\mathcal{T}_i$, the torsion angle $\varphi_{i_3,i_2,i_1,i_\ell}$ can be interpreted as a transformation from the reference frame $\{x_{i_\ell}, x_{i_2}, x_{i_1}\}$ to the frame $\{x_{i_3}, x_{i_2}, x_{i_1}\}$. In this context, the orientation of the angle $\varphi_{i_3,i_2,i_1,i_\ell}$ plays a critical role in defining the direction of this transformation. The expression for the cosine of the torsion angle, provided in \eqref{sec:2:eq:costau}, enables the computation of the absolute value of $\varphi_{i_3,i_2,i_1,i_\ell}$ for all $\ell = 4,\ \dots,\ \kappa_i$; however, it does not yield the orientation. The following lemma addresses this by providing a method to determine the sign of the torsion angles.
	\begin{lemma}
		\label{sec:3:lemma:9}
		Let $x_u, x_v, x_w, x_z \in \mathbb{R}^3$ be such that $x_u,x_v,x_w$ and $x_v,x_w,x_z$ are non-collinear, and let $\varphi_{u,v,w,z}$ be the torsion angle between the planes $\{x_u,x_v,x_w\}$ and $\{x_v,x_w,x_z\}$. If $\varphi_{u,v,w,z}\notin\{0,\pi\}$, then the sign of $\varphi_{u,v,w,z}$ is given by
		\begin{equation*}
			\textnormal{sign}(\varphi_{u,v,w,z}) = \textnormal{sign}\Big(\big[(x_w - x_v) \times (x_u - x_v)\big] \cdot (x_z - x_v)\Big).
			\label{sec:3:eq:lemma:9}
		\end{equation*}
	\end{lemma}
	
	\begin{proof}
		The orientation of the torsion angle $\varphi_{u,v,w,z}$ is defined relative to the plane $\Pi = \{x_u, x_v, x_w\}$ and the position of the point $x_z$. For the nonplanar cases considered in the statement, $\Pi$ partitions the admissible positions of $x_z$ into two open half-spaces, and the sign of $\varphi_{u,v,w,z}$ is determined by the half-space containing $x_z$.
		
		Let $\mathbf{r} = x_u - x_v,\ \mathbf{s} = x_w - x_v$, and denote by $\hat{\mathbf{r}}$ and $\hat{\mathbf{s}}$ their normalized counterparts. Describe $\mathbf{t} = \hat{\mathbf{r}} - (\hat{\mathbf{r}} \cdot \hat{\mathbf{s}})\,\hat{\mathbf{s}}$, so that $\mathbf{t}$ is orthogonal to $\hat{\mathbf{s}}$, and let $\hat{\mathbf{t}}$ denote the corresponding unit vector. The unit normal vector to the plane $\Pi$ is then given by $\hat{\mathbf{n}} = \hat{\mathbf{s}} \times \hat{\mathbf{t}}$, which uniquely determines the orientation of $\Pi$.
		
		Let $x_z^\Pi$ denote the orthogonal projection of $x_z$ onto $\Pi$. Using the standard reduction of affine projections to projections onto linear subspaces, this projection can be expressed as (see, e.g.,~\cite{2016_strang_itla}):
		\begin{equation*}
			x_z^\Pi = \text{proj}_\Pi(x_z) = x_z - \hat{\mathbf{n}} \cdot (x_z - x_v) \hat{\mathbf{n}},
		\end{equation*}
		
		Consequently, the vector $\mathbf{m} = x_z - x_z^\Pi$ is normal to $\Pi$ and collinear with $\hat{\mathbf{n}}$. The sign of the torsion angle $\varphi$ is therefore determined by the relative orientation of $\hat{\mathbf{m}}$ and $\hat{\mathbf{n}}$, and is given by:
		\begin{equation*}
			\begin{array}{lcl}
				\textnormal{sign}(\varphi_{u,v,w,z}) & = & \hat{\mathbf{n}} \cdot \hat{\mathbf{m}} = \hat{\mathbf{n}} \cdot \left[\dfrac{\hat{\mathbf{n}} \cdot (x_z - x_v) \hat{\mathbf{n}}}{\|\mathbf{m}\|}\right] = \dfrac{1}{\|\mathbf{m}\|}\hat{\mathbf{n}} \cdot (x_z - x_v) \\[0.3cm]
				
				& = & \dfrac{1}{\|\mathbf{s}\|\|\mathbf{t}\|\|\mathbf{m}\|} (\mathbf{s} \times \mathbf{t}) \cdot \left(x_z - x_v\right) = \dfrac{1}{\|\mathbf{s}\|\|\mathbf{t}\|\|\mathbf{m}\|} (\mathbf{s} \times \mathbf{r}) \cdot \left(x_z - x_v\right)\\[0.5cm]
				
				& = & \textnormal{sign}\Big(\big[(x_w - x_v) \times (x_u - x_v)\big] \cdot (x_z - x_v)\Big).
			\end{array}
		\end{equation*}
		The formula above applies only when $\varphi_{u,v,w,z}\notin\{0,\pi\}$, since then $x_z$ lies strictly in one of the two open half-spaces determined by $\Pi$. If the scalar triple product vanishes and the two planes are well defined, the configuration is coplanar and the torsion angle is a boundary value, $0$ or $\pi$. In these boundary cases, the scalar triple product does not define an orientation sign; the angle is determined by its absolute value together with the representative convention in $(-\pi,\pi]$.
	\end{proof}
	
	The sequence of results developed in this section has led to a systematic construction of the solution $\mathcal{A}_i$ to the nonlinear system \eqref{sec:2:eq:system_xiI_k>3}. As established by the construction, this solution consists of arcs on a circle in $\mathbb{R}^3$, parameterized by the angular set $\mathcal{T}_i$, which may be represented as a finite union of angular intervals. More generally, we can characterize the solution $\mathcal{A}_i$ of \eqref{sec:2:eq:system_xiI_k>3} as follows:
	\begin{enumerate}
		\item $\mathcal{A}_i = \emptyset \Leftrightarrow \mathcal{T}_i = \emptyset$;
		\item $\mathcal{A}_i$ is a discrete set $\Leftrightarrow \mathcal{T}_i$ is a discrete set;
		\item $\mathcal{A}_i$ is uncountable $\Leftrightarrow \mathcal{T}_i$ is uncountable.
	\end{enumerate}
	
	Based on the reduced interval distances $\mathcal{D}_{i,i_\ell}$ obtained from the original intervals $\mathcal{D}^0_{i,i_\ell}$ for $\ell = 3,\ \dots,\ \kappa_i$, the expression for $\mathcal{T}_i$ in \eqref{sec:3:eq:proposition:4:T_i}, and the three cases above, we conclude:
	\begin{itemize}
		\item $\mathcal{A}_i$ is either empty or a discrete set if there exists $3 \leq \ell \leq \kappa_i$ such that $\mathcal{D}_{i,i_\ell} = d_{i,i_\ell}$;
		\item If $\kappa_i = 3$ and $\mathcal{D}_{i,i_3} \neq d_{i,i_3}$, then $\mathcal{A}_i$ is the union of two circular arcs symmetric with respect to the plane $\{x_{i_3}, x_{i_2}, x_{i_1}\}$ (\Cref{sec:2:proposition:2}). If $\mathcal{D}_{i,i_3} = d_{i,i_3}$, each arc reduces to a single point.
	\end{itemize}
	
	With the result of \Cref{sec:3:proposition:4} in hand, the proposition of a new \textit{i}BP-based algorithm that fully accounts for the characterization of $\mathcal{A}_i$ becomes a natural step. Since the construction of $\mathcal{A}_i$ is based on angular intervals and their intersections, we refer to the newly proposed algorithm as the interval Angular Branch-and-Prune (\textit{i}ABP) method. In this new approach, for each vertex, the set $\mathcal{A}_i$ is computed; if it is nonempty, a sample $\mathbf{A}_i$ is drawn from it, and one point at a time is selected to proceed in the Branch step. Conversely, if $\mathcal{A}_i$ is empty, we are in the Prune phase, meaning that backtracking is necessary to explore an alternative path in the search tree. In practice, $\mathcal{T}_i$ is represented as a finite union of angular interval components in $(-\pi,\pi]$, and intersections are performed componentwise, with empty intersections leading to pruning. The sample $\mathbf{A}_i$ can then be constructed by first selecting a finite subset $T_i=\operatorname{sample}(\mathcal{T}_i)$ from this union of components, where $\mathcal{T}_i$ is defined in \eqref{sec:3:eq:proposition:4:T_i}, and then letting:
	\begin{equation*}
		\mathbf{A}_i = \text{sample}\left(\mathcal{A}_i\right) = \left\{{^{(i_3,i_2,i_1)}\chi}_i(\tau) \mid \tau \in T_i \right\},
		\label{sec:3:eq:A_i}
	\end{equation*}
	
	\noindent where ${^{(u,v,w)}\chi}_z$ denotes the parametric mapping introduced in \Cref{sec:2:lemma:7}, which assigns to each torsion angle a unique point in $\mathbb{R}^3$. Under this construction, every element of $\mathbf{A}_i$ is locally feasible by definition (up to numerical tolerance in practice) with respect to the angular constraints encoded in $\mathcal{T}_i$. Thus, the branching procedure reduces the explicit feasibility checks required in the standard \textit{i}BP algorithm, while the numerical feasibility of generated realizations is still assessed in the computational workflow.
	
	\autoref{sec:3:fig:Ai_iABP} presents an example illustrating the construction of a nonempty set $\mathbf{A}_i$ within the \textit{i}ABP algorithm. In this example, $\kappa_i = 4$, $|T_i^+| = 5$, and $|T_i^-| = 0$, since $\mathcal{T}_i^+ \neq \emptyset$ whereas $\mathcal{T}_i^- = \emptyset$. The set $T_i = T_i^- \cup T_i^+$ is obtained by uniformly sampling $\mathcal{T}_i$ at equally spaced points. This corresponds to the same configuration depicted in \autoref{sec:2:fig:Ai_iBP}, and a comparison reveals that the number of feasible sample points at this level is greater under the \textit{i}ABP approach. When comparing the two methods, it is important to emphasize that the computational cost of placing a vertex with \textit{i}ABP is generally higher than with \textit{i}BP, since \textit{i}ABP requires the explicit construction of the entire set $\mathcal{A}_i$, a step not necessary in \textit{i}BP. Nevertheless, in \textit{i}ABP, every point representing an embedding of the vertex $v_i$ is feasible with respect to $\mathcal{A}_i$ by construction -- within numerical tolerance -- whereas in \textit{i}BP, each candidate point must undergo an explicit feasibility check.
	\begin{figure}[!htp]
		\centering
		\begin{subfigure}[b]{0.49\linewidth}
			\centering
			\includegraphics[width=0.9\linewidth]{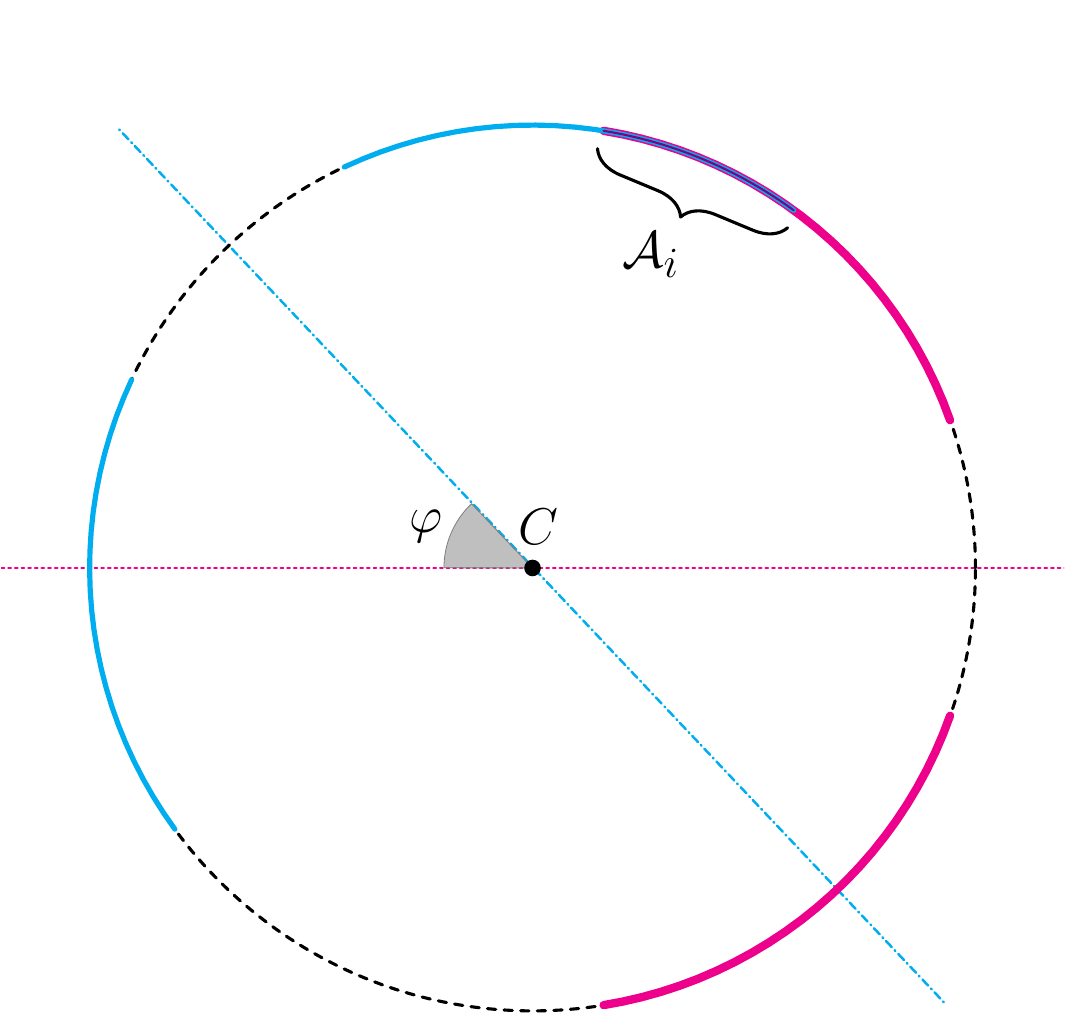}
			\caption{$\mathcal{A}_i = {^{(3,2,1)}\mathcal{A}}_i \cap {^{(4,2,1)}\mathcal{A}}_i$}
		\end{subfigure}
		\begin{subfigure}[b]{0.49\linewidth}
			\centering
			\includegraphics[width=0.9\linewidth]{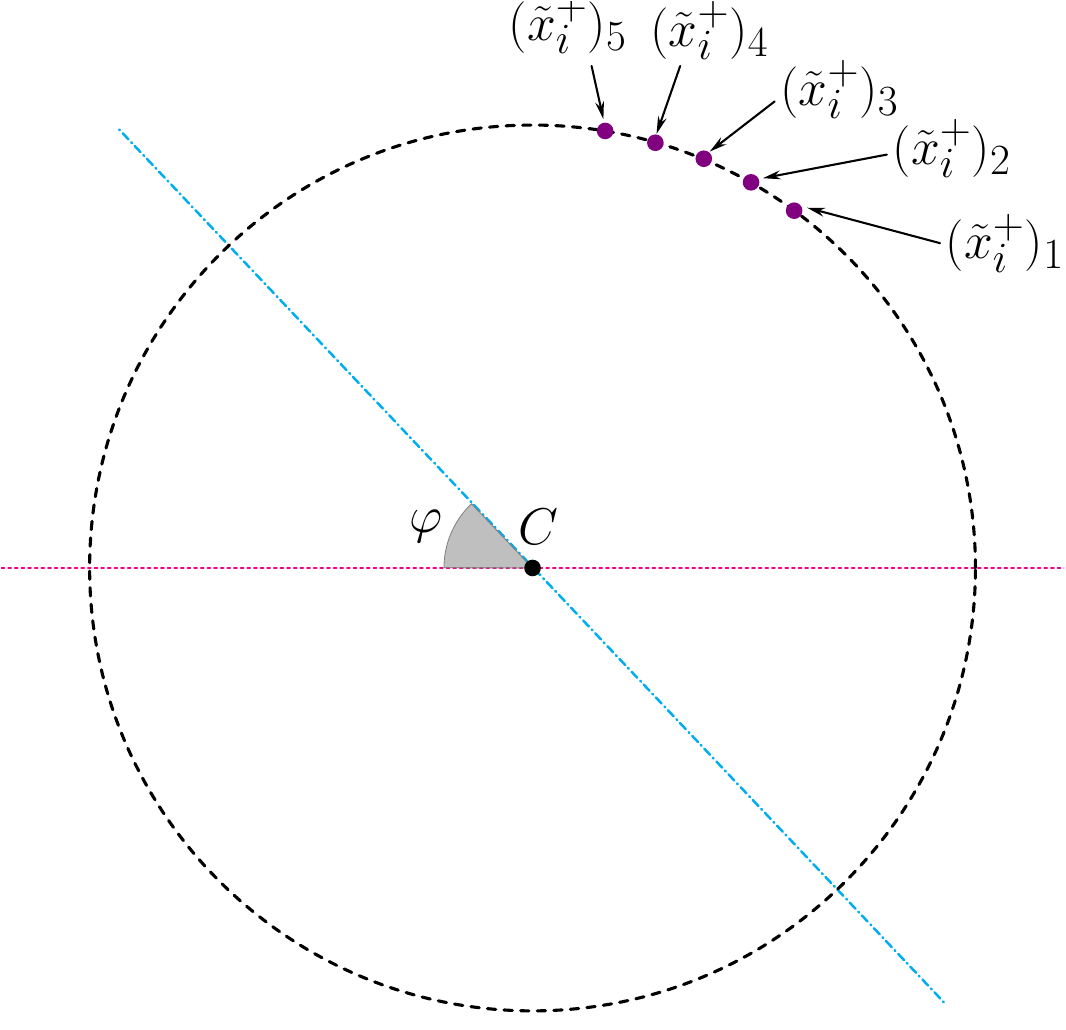}
			\caption{$\mathbf{A}_i = \text{sample}(\mathcal{A}_i)$\phantom{$^{(}$}}
		\end{subfigure}
		
		\caption{View of the plane containing the intersection $S_i^1 \cap S_i^2$, represented as a dashed black circumference, where the angle $\varphi$ denotes the torsion between the planes $\{x_{i_4}, x_{i_2}, x_{i_1}\}$ and $\{x_{i_3}, x_{i_2}, x_{i_1}\}$. In (a), the magenta arcs correspond to the intersection of $\mathcal{S}_i^3$ with this fixed circumference, yielding the set ${^{(3,2,1)}\mathcal{A}}_i$, while the cyan arcs correspond to the intersection of $\mathcal{S}_i^4$ with the same circumference, thereby yielding ${^{(4,2,1)}\mathcal{A}}_i$. The purple arc represents the refined set $\mathcal{A}_i = {^{(3,2,1)}\mathcal{A}}_i \cap {^{(4,2,1)}\mathcal{A}}_i$. In (b), the sampled points from $\mathcal{A}_i$ are $\mathbf{A}_i = \big\{\big(\tilde{x}_i^+\big)_1, \ \big(\tilde{x}_i^+\big)_2, \ \big(\tilde{x}_i^+\big)_3, \ \big(\tilde{x}_i^+\big)_4, \ \big(\tilde{x}_i^+\big)_5\big\}$.}
		\label{sec:3:fig:Ai_iABP}
	\end{figure}
	
	For all torsion angles constructed solely from distance information, for example with \eqref{sec:2:eq:costau}, it is important to recall that such information alone is insufficient to determine their sign. However, in applications involving molecular bio-structures, the orientation of certain torsion angles may be known \textit{a priori} for specific atoms \cite{LPB_2021}. To accommodate this additional information, we introduce the interval Torsion Angle Branch-and-Prune (\textit{i}TBP) method. This algorithm retains the same search structure as \textit{i}ABP, but allows prescribed input torsion-angle intervals to further restrict the initial angular set when such intervals are available. More precisely, $\mathcal{T}_i^0$ is first computed from the reduced distance interval $\mathcal{D}_{i,i_3}$ as in \eqref{sec:2:eq:T_i^0}; if an input torsion interval $\mathcal{T}_i^{\mathrm{in}}$ is supplied, we then replace $\mathcal{T}_i^0$ by $\mathcal{T}_i^0\cap\mathcal{T}_i^{\mathrm{in}}$. The final angular set sampled by \textit{i}TBP is computed from \eqref{sec:3:eq:proposition:4:T_i}. Thus, with this convention, $\textit{i}\mathrm{ABP}=\textit{i}\mathrm{TBP}(\mathcal{T}_i^{\mathrm{in}}=\texttt{NONE},\ i=4,\ \dots,\ |V|)$, where $\texttt{NONE}$ indicates that no additional torsion-angle restriction is imposed at the corresponding level. A pseudocode summarizing the procedure introduced in this section is provided in \Cref{sec:3:pc:iTBP}, and computational results comparing the performance of both algorithms are presented in \Cref{sec:5}.
	
\begin{algorithm}[!htp]
	\caption{A main framework of the \textit{i}TBP algorithm in $\mathbb{R}^3$.}
	\label{sec:3:pc:iTBP}
	\begin{algorithmic}[1]\small
		\Require{$G=(V,E,\mathcal{D})$, $X=\{x_1,x_2,x_3\}$, predecessor sets $U=[U_4,\ \dots,\ U_{|V|}]$ with $U_i=\{i_1,i_2,i_3,\dots,i_{\kappa_i}\}$, sampling sizes $N=[\nu_4,\ \dots,\ \nu_{|V|}]$ with $\nu_i\geq 1$, and optional input torsion angle intervals $\mathcal{T}^{\mathrm{in}}=[\mathcal{T}_4^{\mathrm{in}},\dots,\mathcal{T}_{|V|}^{\mathrm{in}}]$, where $\mathcal{T}_i^{\mathrm{in}}=\texttt{NONE}$ means no input torsion angle interval is supplied}
		\Ensure{A feasible sampled realization $X$, or $\emptyset$ if no feasible realization is found in the finite sampled search}

		\Function{$\mathit{i}\mathrm{TBP}$}{$G,X,U,N,\mathcal{T}^{\mathrm{in}}$}
		\State $i\leftarrow 4$;
		\For{$j=4,\dots,|V|$}
			\State $T_j\leftarrow\emptyset$ and $r_j\leftarrow0$;
		\EndFor
		\While{$i\geq 4$}
			\If{$i=|V|+1$}
				\State \Return $X$;
			\ElsIf{$r_i=0$}
				\State Compute the reduced interval $\mathcal{D}_{i,i_3}$ and the angular set $\mathcal{T}_i^0$ as in \eqref{sec:2:eq:T_i^0};
				\If{$\mathcal{T}_i^{\mathrm{in}} \neq \texttt{NONE}$}
					\State $\mathcal{T}_i^0 \leftarrow \mathcal{T}_i^0 \cap \mathcal{T}_i^{\mathrm{in}}$;
				\EndIf
				\State Compute $\mathcal{T}_i$ from \eqref{sec:3:eq:proposition:4:T_i}, using circular interval intersections in $(-\pi,\pi]$;
				\If{$\mathcal{T}_i=\emptyset$}
					\State $T_i\leftarrow\emptyset$;
				\Else
					\State Compute a finite sample $T_i\subset\mathcal{T}_i$ satisfying $|T_i|\leq \nu_i$;
					\State Set $\mathbf{a}_i\leftarrow\mathbf{a}(i_2,i_1,i)$, $\mathbf{b}_i\leftarrow\mathbf{b}(i_3,i_2,i_1,i)$, and $\mathbf{c}_i\leftarrow\mathbf{c}(i_3,i_2,i_1,i)$ as in \eqref{sec:2:eq:a_b_c};
				\EndIf
				\State $r_i\leftarrow1$;
			\ElsIf{$T_i=\emptyset$ or $r_i>|T_i|$}
				\State $T_i\leftarrow\emptyset$ and $r_i\leftarrow0$;
				\State $i\leftarrow i-1$;
				\If{$i\geq4$}
					\State $X\leftarrow X\setminus\{x_i\}$ and $r_i\leftarrow r_i+1$;
				\EndIf
			\Else
				\State $\tau\leftarrow$ the $r_i$-th torsion angle in $T_i$;
				\State $x_i\leftarrow \mathbf{a}_i+\mathbf{b}_i\cos\tau+\mathbf{c}_i\sin\tau$;
				\State $X\leftarrow X\cup\{x_i\}$;
				\State $i\leftarrow i+1$;
				\If{$i\leq |V|$}
					\State $T_i\leftarrow\emptyset$ and $r_i\leftarrow0$;
				\EndIf
			\EndIf
		\EndWhile
		\State \Return $\emptyset$;
		\EndFunction
	\end{algorithmic}
\end{algorithm}
	
	\section{Generating \textit{i}DDGP Instances from Protein Structures}\label{sec:4}
	
	\subsection{Protein Chain Modeling}
	\label{sec:4:1}
	
	In a simplified framework, proteins can be described as linear chains of amino acid residues linked by peptide bonds \cite{TBCoNAaP_2010}. The protein backbone is represented by a sequence of atom triplets $\{N^i,\ C_\alpha^i,\ C^i\}$, where $i = 1,\ \dots,\ n_{\text{aa}}$, and $n_{\text{aa}}$ denotes the number of amino acids in the protein. The spatial configuration of the backbone is determined by torsion angles defined along the chain. \Cref{sec:4:fig:protein_overview} illustrates the bonds associated with the $i$-th residue ($R_i$) and its immediately adjacent residues, highlighting also the backbone torsion angles $\omega_i,\ \phi_i$, and $\psi_i$, defined as follows \cite{LPB_2021}:
	\begin{equation*}
		\phi_i := C^{i-1} - N^i - C_\alpha^i - C^i,\ \psi_i := N^i - C_\alpha^i - C^i - N^{i+1}, \text{ and } \omega_i := C_\alpha^{i-1} - C^{i-1} - N^i - C_\alpha^i,
		\label{sec:4:eq:PhiPsiOmega}
	\end{equation*}
	
	\noindent where the $\omega$ angle typically assumes two characteristic values: $\omega \approx 0^{\circ}$ for a \textit{cis} peptide bond, and $\omega \approx 180^{\circ}$ for a \textit{trans} peptide bond.
	
	\begin{figure}[!htp]
		\centering
		\includegraphics[width=0.9\textwidth]{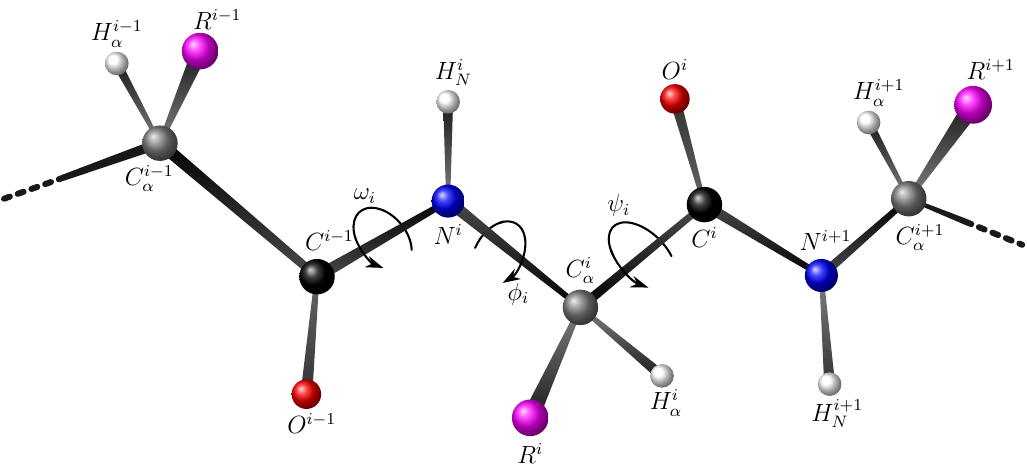}
		\caption{Representation of the $i$-th residue of a protein and its local adjacency, highlighting all relevant atoms and the torsion angles $\omega_i,\ \phi_i$, and $\psi_i$.}
		\label{sec:4:fig:protein_overview}
	\end{figure}
	
	Our focus will be restricted to determining the positions of backbone atoms. Once these atoms are fixed, the side chain atoms can subsequently be placed using strategies such as the one proposed in \cite{DOfPSC_2014}. Thus, our primary objective is to determine the spatial configuration of the backbone atoms.
	
	NMR experiments provide estimates of interatomic distances between atoms that are spatially proximal within protein molecules. To incorporate this information and enhance triangulation accuracy, we include backbone-attached hydrogen atoms in the set of considered vertices. Accordingly, for a protein chain with ${n_{\text{aa}}}$ amino acid residues, the relevant atom set consists of the backbone atoms $\{N^i,\, C_\alpha^i,\, C^i \mid i = 1,\dots,n_{\text{aa}}\}$ together with selected hydrogen atoms $\{H_3,\, H_2,\, H_N^1\} \;\cup\; \{H_\varrho^i,\, H_\rho^i \mid i = 1,\dots,n_{\text{aa}}\}$.
	
	Since proline residues lack a hydrogen atom ($H_N$) attached to the nitrogen ($N$) atom \cite{LPB_2021}, we consider that this missing atom is replaced by one of the hydrogen atoms attached to the $\delta$-carbon ($C_\delta$), as done in \cite{2023_PDAiPPaNDotRM}. Thus, all hydrogen atoms attached to nitrogen atoms are labeled as $H_\varrho^i$, where $\varrho = \delta_2$ if $R_i$ is proline, and $\varrho = N$ otherwise. For glycine residues, which possess two $H_\alpha$ atoms attached to $C_\alpha$ \cite{LPB_2021}, we consider only one, denoted as $\alpha_2$, i.e., $\rho = \alpha_2$ if $R_i$ is glycine, and $\rho = \alpha$ otherwise.
	
	The following assumptions are adopted to define the set of distance restraints:
	\begin{enumerate}
		\item The amino acid sequence of the protein is assumed to be known \textit{a priori}~\cite{2011_donald_aismb};
		
		\item All covalent bond lengths and bond angles are experimentally determined, with the latter strictly ranging between $0^{\circ}$ and $180^{\circ}$~\cite{ABaAPfXrPSR_1991}. As a consequence, the following pairwise distances are assumed to be exactly known: all distances among atoms $H_3$, $H_2$, $H_N^1$, and $N_1$; all distances among $N^i$, $C_\alpha^i$, $C^i$, and $H^i_\rho$ for $i = 1,\ \dots,\ {n_{\text{aa}}}$; all distances from $C^{i-1}$ to $N^i$ and $H_N^i$ for $i = 2,\ \dots,\ {n_{\text{aa}}}$; and distances from $N^i$ to $C_\alpha^{i-1}$, $C^{i-1}$, and $H_{\delta_2}^i$ for $i = 2,\ \dots,\ {n_{\text{aa}}}$ (see~\autoref{sec:4:fig:DDGP_hcOrder}); \label{sec:4:precisedistances}
		
		\item Under the idealized planarity assumption~\cite{LPB_2021}, all distances between atoms within the same extended peptide group ($C_\alpha^{i}$, $C^{i}$, $O^{i}$, $N^{i+1}$, $H_\varrho^{i+1}$, $C_\alpha^{i+1}$) are considered known (see~\autoref{sec:4:fig:protein_overview}). Nevertheless, a recent study~\cite{2024_darocha_IoSiaLAfCPC} has shown that even small deviations from the ideal planarity of the peptide bond can significantly affect protein conformations computed using \textit{i}BP. As a consequence of these two considerations, the distances from $C^i_\alpha$ to $C^{i+1}_\alpha$, $H_N^{i+1}$, and $H_{\delta_2}^{i+1}$, as well as from $H_{\delta_2}^{i+1}$ to $C^i$ and $C_\alpha^{i+1}$, are assumed to be known for all $i = 1,\ \dots,\ {n_{\text{aa}}}-1$, along with all torsion angles defined by any quadruple of atoms within the extended peptide group, including not only their absolute values but also their orientation;
		\label{sec:4:pgdistances}
		
		\item Interatomic distances between hydrogen atoms separated by less than $5\,\text{\AA}$ can be estimated from NOESY NMR experiments~\cite{wuthrich1986}. The corresponding interval restraints provided by such measurements include: $\mathcal{D}(H_2, H_\rho^1)$, $\mathcal{D}(H_3, H_\rho^1)$, $\mathcal{D}(H_N^1, H_\rho^1)$, $\mathcal{D}(H_\rho^{i-1}, H_\varrho^i)$, and $\mathcal{D}(H_\varrho^i, H_\rho^i)$ for $i = 2,\ \dots,\ {n_{\text{aa}}}$;
		\label{sec:4:nmrdistances}
		
		\item Estimates for the torsion angles $\phi$ and $\psi$ are available, for example, from TALOS-N, a neural network-based method trained on chemical shift data~\cite{shen2015_talosn}. Accordingly, we assume the interval distances $\mathcal{D}(N^{i-1}, N^i)$ and $\mathcal{D}(C^{i-1}, C^i)$ are available for all $i = 2,\ \dots,\ {n_{\text{aa}}}$ (see~\autoref{sec:4:fig:protein_overview});
		\label{sec:4:torsion}
		
		\item The chirality, i.e., the spatial orientation of a tetrahedron, is known for the set of atoms $\{N,\ C_\alpha,\ C,\ H_\alpha\}$ and for any tetrahedron defined by four atoms in $\{H^3,\ H^2,\ H_N^1,\ N^1,\ C_\alpha^1\}$~\cite{LPB_2021};
		\label{sec:4:chirality}
		
		\item The minimum allowed distance between any two atoms is given by the sum of their van der Waals radii~\cite{1964_bondi_vdwvar}, and is represented as an interval constraint of the form $\big[\underline{d}, \infty\big)$.\label{sec:4:vdwt}
	\end{enumerate}
	
	For lower-bound-only restraints of the form $[\underline{d},\infty)$, the effective interval used in the angular construction is obtained by intersecting this restraint with the corresponding reachable distance interval; hence the associated angular set is still represented in $(-\pi,\pi]$.
	
	\subsection{Construction of a Vertex Order for the \textit{i}DDGP}
	
	The work in \cite{OtCoPBbUABoH_2011} first modeled protein structures using a DDGP vertex order, thereby anticipating the hand-crafted (\textit{hc}) order later formalized in \cite{2013_lavor_tibpaftdmdgpwid}. Subsequent contributions \cite{APDBGRfDOfDG_2015, 2017_goncalvez_raotidgp, 2019_lavor_mndifropg} proposed alternative vertex orderings. In the latter work, the oxygen atom was included among the considered vertices, and the ordering followed an \textit{i}DMDGP scheme with repeated vertices. Starting from this ordering, we remove the oxygen atoms and eliminate repeated vertices, thereby obtaining a modified \textit{hc} order that preserves the branching structure associated with hydrogen atoms. We emphasize that this order is not unique. The resulting order is illustrated in \autoref{sec:4:fig:DDGP_hcOrder} for ${n_{\text{aa}}} = 4$ and formally specified below.
	
	\begin{align*}
		hc = \, &\big\{H_3,\ H_2,\ H_N^1,\ N^1,\ C_\alpha^1,\ H_\rho^1,\ C^1,\ H_\varrho^2,\ C_\alpha^2,\ N^2,\ H_\alpha^2,\ C^2,\ \dots,\\
		& \phantom{\big\{}H_\varrho^i,\ C_\alpha^i,\ N^i,\ H_\alpha^i,\ C^i,\ \dots, H_\varrho^{n_{\text{aa}}},\ C_\alpha^{n_{\text{aa}}},\ N^{n_{\text{aa}}},\ H_\alpha^{n_{\text{aa}}},\ C^{n_{\text{aa}}}\big\}.
	\end{align*}
	
	\begin{figure}[!htp]
		\centering
		\includegraphics[width=0.9\linewidth]{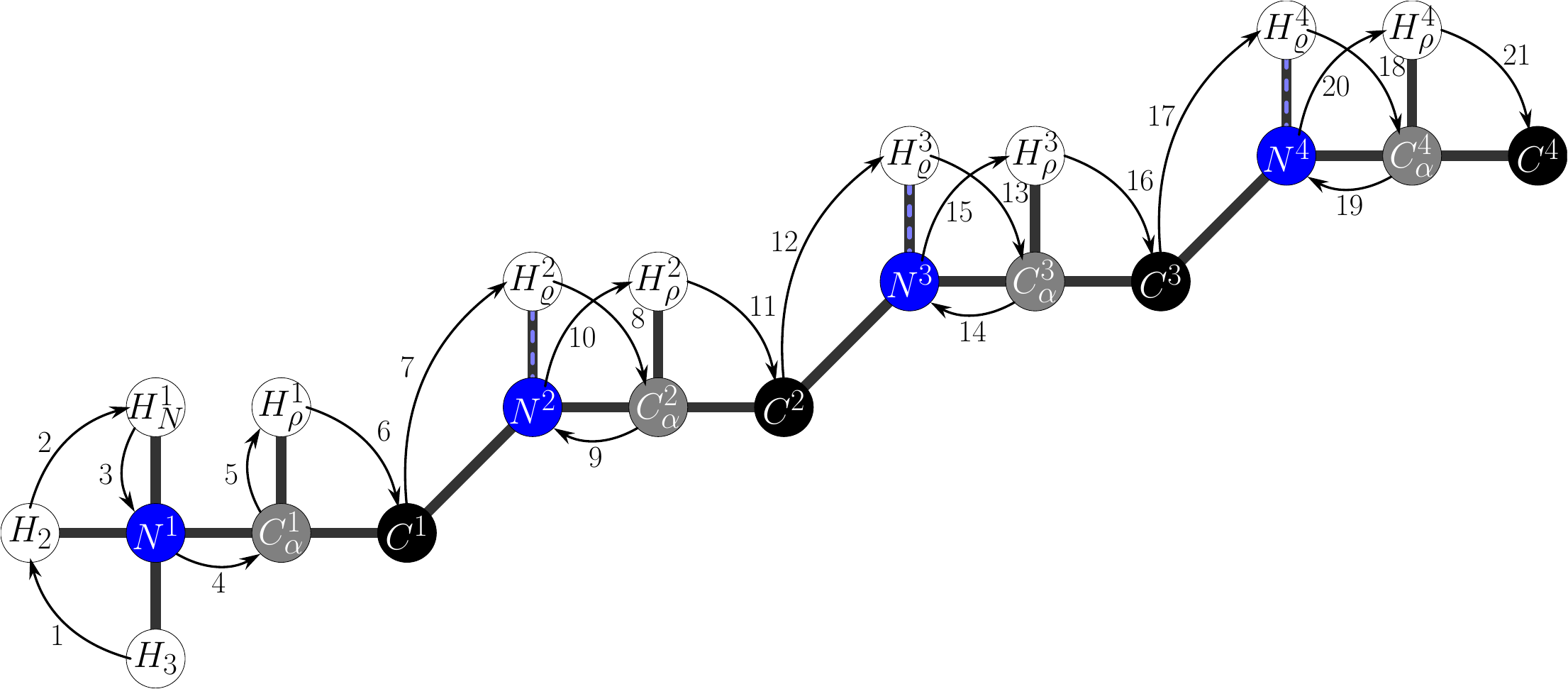}
		\caption{\textit{hc} order for the considered atoms of a protein composed of 4 amino acid residues.}
		\label{sec:4:fig:DDGP_hcOrder}
	\end{figure}
	
	Let us define a protein graph $G = (V, E, \mathcal{D})$ associated with the considered protein atoms. The vertices in $V$ are labeled according to the following procedure:
	\begin{equation}
		\begin{array}{ccl}
			1^{st}\text{ res.} & : &
			v_1 \leftarrow H^{3},\
			v_2 \leftarrow H^{2},\
			v_3 \leftarrow H_N^1,\
			v_4 \leftarrow N^1,\
			v_5 \leftarrow C_\alpha^1,\
			v_6 \leftarrow H_\rho^1,\
			v_7 \leftarrow C^1,\\
			\vdots & & \hspace{4cm}\vdots\\
			i^{th}\text{ res.} & : &
			v_{5i-2} \leftarrow H_\varrho^i,\
			v_{5i-1} \leftarrow C_\alpha^i,\
			v_{5i} \leftarrow N^i,\
			v_{5i+1} \leftarrow H_\rho^i,\
			v_{5i+2} \leftarrow C^i,\\[0.1cm]
			&&2\leq i \leq {n_{\text{aa}}}.
		\end{array}
		\label{sec:4:eq:order_hc}
	\end{equation}
	
	The edges in $E$ are established based on the existence of the distances described in \Cref{sec:4:1}, and the weight $\mathcal{D}$ of each edge corresponds to the distance associated with the respective pair of atoms. Based on this, we state the following proposition:
	\begin{proposition}
		\label{sec:4:proposition:5}
		The \textit{hc} order is an \textit{i}DDGP order.
	\end{proposition}
	
	\begin{proof}
		To prove this proposition, we need to verify the following three conditions: (I) the subgraph induced by $\{v_1,\ v_2,\ v_3\}$ is a clique in $E_0$, and the associated weights satisfy the strict triangle inequalities; (II) for each vertex $v_j$, there exists a subset $U_{v_j}^0 := \{v_{j_3},\ v_{j_2},\ v_{j_1}\} \subset U_{v_j}$ such that $\big\{\{v_j, v_{j_2}\},\ \{v_j, v_{j_1}\}\big\} \subset E_0$; and (III) the strict triangle inequality holds for the edge weights among the vertices of $U_{v_j}^0$.
		
		The table below presents a possible clique $\{v_{j_3}, v_{j_2}, v_{j_1}, v_j\}$, with $3 \leq j \leq |V|$, for each atom of each residue $i$ in the \textit{hc} order, for $i = 1,\ 2,\ \dots,\ {n_{\text{aa}}}$, ensuring that each vertex has at least three vertices adjacent predecessors.
		
		\begin{center}
\begin{tabular}{@{}c|c|c|c|c|@{}|c|c|c|c|c@{}}
	& {$v_{j_3}$} & {$v_{j_2}$} & {$v_{j_1}$} & {$v_j$}
	&  & {$v_{j_3}$} & {$v_{j_2}$} & {$v_{j_1}$} & {$v_j$}\\
	\thickhline

	\multirow{5}{*}{\rotatebox{90}{Residue $1$}}
	& {$-$} & {$H^3$} & {$H^2$} & {$H_N^1$}
	& \multirow{5}{*}{\rotatebox{90}{Residue $i$}}
	& {$H_\rho^{i-1}$} & {$C_\alpha^{i-1}$} & {$C^{i-1}$} & {$H_\varrho^i$}\\

	& {$H_3$} & {$H_2$} & {$H_N^1$} & {$N^1$}
	&
	& {$C_\alpha^{i-1}$} & {$C^{i-1}$} & {$H_\varrho^i$} & {$C_\alpha^i$}\\

	& {$H_2$} & {$H_N^1$} & {$N^1$} & {$C_\alpha^1$}
	&
	& {$C_\alpha^{i-1}$} & {$C^{i-1}$} & {$C_\alpha^{i}$} & {$N^i$}\\

	& {$H_N^1$} & {$N^1$} & {$C_\alpha^1$} & {$H_\alpha^1$}
	&
	& {$H_\varrho^i$} & {$N^i$} & {$C_\alpha^i$} & {$H_\rho^i$}\\

	& {$N^1$} & {$C_\alpha^1$} & {$H_\alpha^1$} & {$C^1$}
	&
	& {$N^{i}$} & {$C_\alpha^i$} & {$H_\rho^i$} & {$C^i$}\\

	\bottomrule
\end{tabular}
		\end{center}
		
		Analyzing the table, we conclude that the subgraph induced by $\{v_1,\ v_2,\ v_3\}$ is a clique in $E_0$, as shown in its first line. It is also shown that $v_{j_3},\ v_{j_2},\ v_{j_1} \in U_{v_j}$, with $d_{j,j_2}$ and $d_{j,j_1}$ being exact distances, thereby verifying condition (II). To prove conditions (I) and (III), we must verify the strict triangle inequality for each clique presented in the table, which can be readily done by applying the law of cosines \cite{2005_strang_laaia} to the corresponding edge weights. Since the proof is analogous in all cases, we detail it only for the clique corresponding to the $6^{\text{th}}$ line.
		
		Considering the relative positions of the atoms along the protein chain (as illustrated in \autoref{sec:4:fig:DDGP_hcOrder}), we observe that the atom represented by $v_{j_2}$ is connected by bond lengths to the atoms represented by $v_{j_1}$ and $v_{j_3}$. Therefore, the angle $\theta_{j_3,j_2,j_1}$ encodes distance information associated with bond angles, and the corresponding distance is $d_{j_1,j_3}$. However, since bond angles are assumed to be strictly between $0$ and $\pi$ in Assumption~\ref{sec:4:precisedistances}, it follows that $\theta_{j_3,j_2,j_1} \in (0,\ \pi)$. Thus, we can write
		\begin{align*}
			d_{j_1,j_3}^2 & = d_{j_1,j_2}^2 + d_{j_2,j_3}^2 - 2d_{j_2,j_3}d_{j_1,j_2}\cos\theta_{j_3,j_2,j_1} \\[0.1cm]
			& = (d_{j_1,j_2} + d_{j_2,j_3})^2 - 2d_{j_2,j_3}d_{j_1,j_2}(1 + \cos\theta_{j_3,j_2,j_1})\\[0.1cm]
			& < (d_{j_1,j_2} + d_{j_2,j_3})^2 \quad\Leftrightarrow \quad d_{j_1,j_3} < d_{j_1,j_2} + d_{j_2,j_3}.
			\label{sec:4:eq:sti}
		\end{align*}
		The remaining strict triangle inequalities follow analogously from the same law-of-cosines argument, since the angle is nonflat. Hence, the distances $d_{j_1,j_2},\ d_{j_2,j_3}$, and $d_{j_1,j_3}$ form a non-degenerate triangle. Consequently, by applying the same reasoning to the remaining cliques in the table, 
		we conclude that conditions (I) and (III) are fulfilled.
	\end{proof}
	
	We also note that the nondegeneracy assumption introduced in \Cref{sec:2} is satisfied by the protein instances constructed with the \textit{hc} order. Indeed, for each row of the table above, the two exact distances used to place $v_j$, namely $d_{j,j_2}$ and $d_{j,j_1}$, together with the exact predecessor distance $d_{j_1,j_2}$, are determined by covalent geometry or by the planar peptide-group geometry described in \Cref{sec:4:1}. The corresponding bond angles are strictly between $0$ and $\pi$; hence, by the same law-of-cosines argument used in the proof of \Cref{sec:4:proposition:5}, the distances $d_{j,j_1}$, $d_{j,j_2}$, and $d_{j_1,j_2}$ satisfy the strict triangle inequalities. Therefore, the two spheres centered at $x_{j_1}$ and $x_{j_2}$ with radii $d_{j,j_1}$ and $d_{j,j_2}$ intersect, whenever feasible, in a nondegenerate circle rather than in a single point.
	
	Based on the result described in the previous proposition and using all distance constraint information provided by our set of assumptions, \autoref{sec:4:fig:full_tree} illustrates half of the solution space search tree corresponding to the $i$-th amino acid residue in the \textit{hc} order.
	\begin{figure}[!htp]
		\centering
		\includegraphics[width=0.9\textwidth]{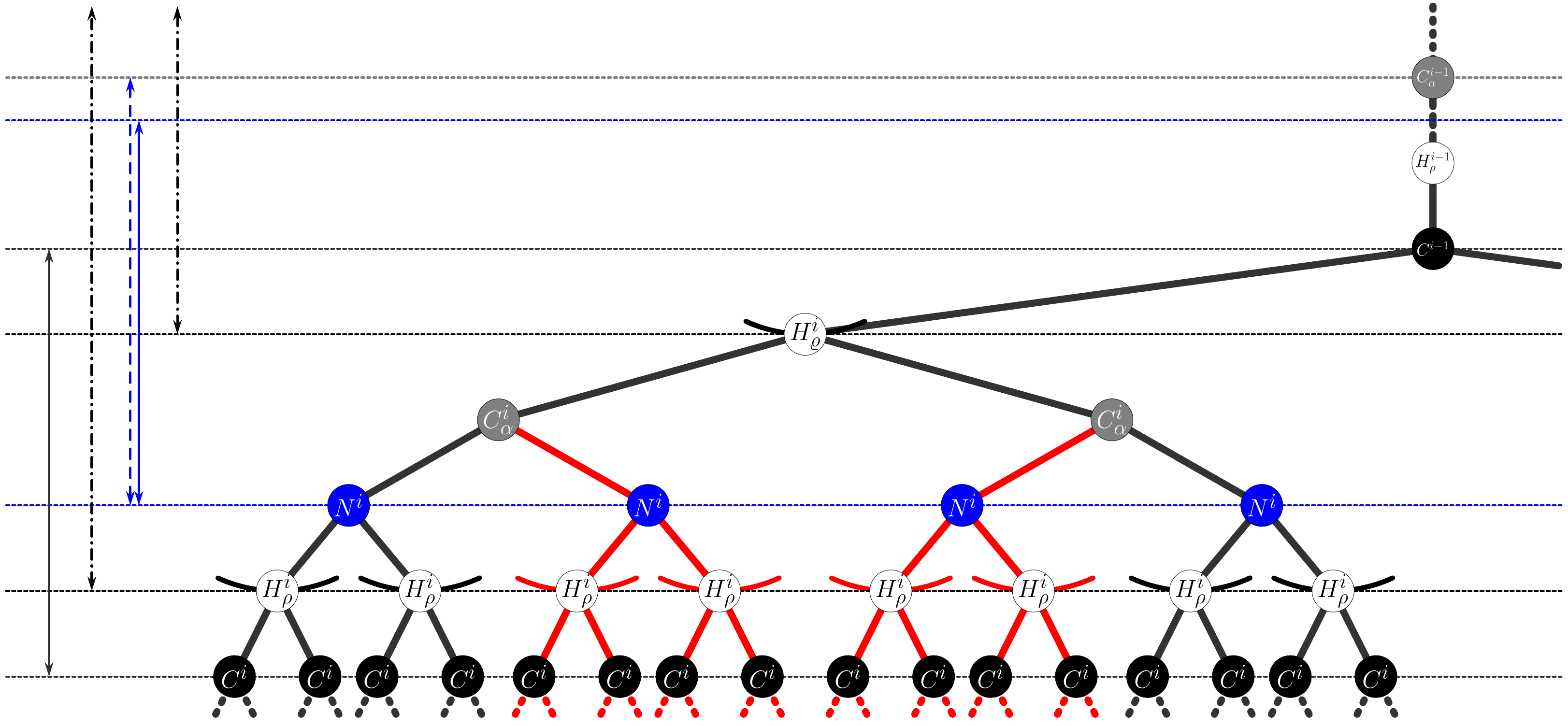}
		\caption{Half of the solution space search tree corresponding to the $i$-th amino acid residue in the \textit{hc} order. The circular arcs represent infinitely many possible placements for atoms located along them. Arrows indicate the presence of distance constraints not used during the discretization process of the \textit{i}DDGP order. Dashed lines correspond to exact distances; solid lines indicate interval distances; and dash-dotted lines represent possible interval distances inferred from NMR experiments. Branches marked in red correspond to portions of the tree where a pruning operation will necessarily occur, due to the existence of an exact distance constraint involving the vertex representing atom $N^i$ (the chosen prune direction is illustrative).}
		\label{sec:4:fig:full_tree}
	\end{figure}
	
	As described in \cite{2017_goncalvez_raotidgp, 2019_lavor_mndifropg}, information about local chirality and the planarity condition imposed by the peptide bond, corresponding to Assumptions \ref{sec:4:chirality} and \ref{sec:4:pgdistances} of the previous subsection, can be used to reduce the search space. However, this information is not available as distance constraints but rather as torsion angles. Since the \textit{i}TBP algorithm is able to integrate torsion angle information, its search space is equivalent to that presented in \autoref{sec:4:fig:tree_LCI}. In the idealized residue-level branching pattern represented in \autoref{sec:4:fig:tree_LCI}, this corresponds to one-fourth of the number of possible paths shown for the unconstrained chirality/planarity representation in \autoref{sec:4:fig:full_tree}. This reduction in the search space is expected to affect both the search behavior and the structural properties of the realizations obtained. These aspects are discussed in \Cref{sec:5}.
	\begin{figure}[!htp]
		\centering
		\includegraphics[width=0.7\textwidth]{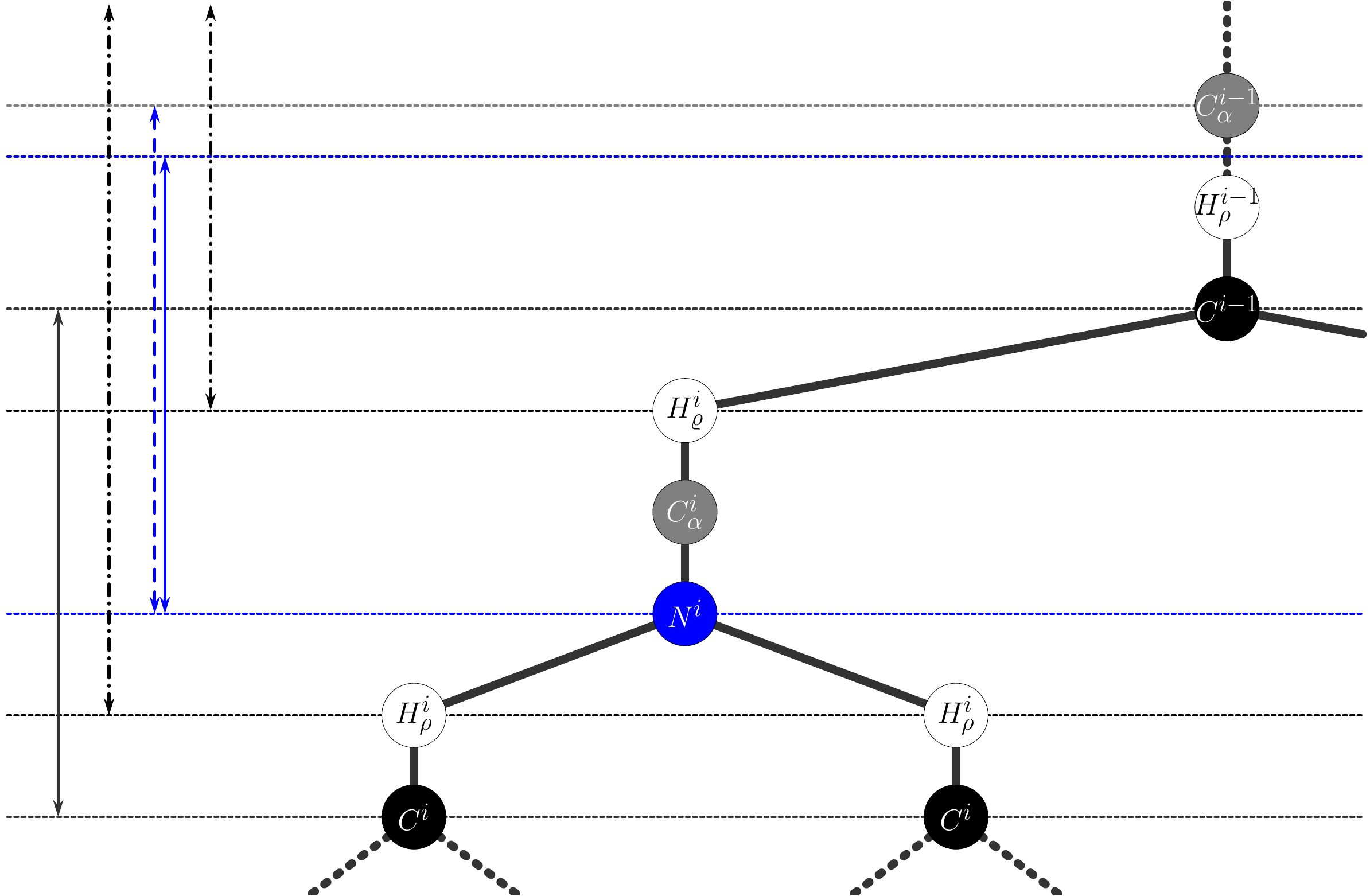}
		\caption{Search tree from \autoref{sec:4:fig:full_tree} incorporating local chirality and peptide bond planarity information.}
		\label{sec:4:fig:tree_LCI}
	\end{figure}
	
	\subsection{Instance Generation}
	\label{subsec:instance_generation}
	
	The problem instances were derived from protein structures available in the RCSB Protein Data Bank (PDB) \cite{TPDB_2000}. For structures containing multiple models, only the first model was considered; for proteins composed of multiple chains, only chain~A was used. Let $G = (V, E, \mathcal{D})$ be the protein graph associated with the selected atoms. For each PDB structure, we assume that all distance constraints described in \Cref{sec:4:1} are known.
	
	\medskip
	\noindent \textbf{Exact distance constraints.}
	Covalent bond lengths and planar bond angles were assumed to be known exactly (see Assumption~\ref{sec:4:precisedistances}) and extracted from the PDB coordinates. As the formulation is expressed solely in terms of distances, planar angles were converted into distances between atom pairs separated by two covalent bonds via the law of cosines.
	
	\medskip
	\noindent \textbf{NMR-derived distance intervals.}
	Interval constraints were defined only for hydrogen atom pairs, as described in Assumption~\ref{sec:4:nmrdistances}. These intervals were modeled as $\mathcal{D}_{i,j} = \big[\underline{d}_{i,j}, \overline{d}_{i,j}\big]$, with
	\begin{equation*}
		\underline{d}_{i,j} = \max\big\{d_{i,j}^* - \varepsilon_{d_{i,j}}/2,\ \mathrm{vdW}_t(H_i, H_j)\big\}, \qquad
		\overline{d}_{i,j} = \min\big\{d_{i,j}^* + \varepsilon_{d_{i,j}}/2,\ \mathrm{cutoff}\big\},
	\end{equation*}
	
	\noindent where $d_{i,j}^*$ given by the normal distribution $\mathcal{N}(\mu,\sigma^2)$, with $\mu = d_{i,j}$ denoting the true distance computed from the PDB structure and $\sigma = \varepsilon_{d_{i,j}}/8$. Hence, $\varepsilon_{d_{i,j}} = 8\sigma \geq 0$ corresponds to an interval of $\pm 4\sigma$ around the mean. The parameter $\mathrm{vdW}_t$ denotes a van der Waals lower bound (see below), while $\mathrm{cutoff}$ represents the maximum observable distance in NMR experiments; in this work, we fix it to $5\,\text{\AA}$. This strategy was adapted from the method proposed in \cite{ADSAfLSNAFGRwAtMC_2008}. The values of $\varepsilon_{d_{i,j}}$ (in \text{\AA}) are defined as
	\begin{equation*}
		\varepsilon_{d_{i,j}} =
		\begin{cases}
			\varepsilon_{\mathrm{short}}  & \text{if atoms } i \text{ and } j \text{ belong to the same or adjacent residues},\\[0.2cm]
			\varepsilon_{\mathrm{long}} & \text{otherwise},
		\end{cases}
		\label{sec:4:eq:Epsd}
	\end{equation*}
	
	\noindent where $(\varepsilon_{\mathrm{short}},\varepsilon_{\mathrm{long}})$ are chosen based on parameter ranges suggested in \cite{2017_goncalvez_raotidgp, 2017_PSbN}, yielding the set of possible combinations $\big\{(0.1,0.5),\ (0.5,1.0),\ (1.0,2.0),\ (1.0,3.0),\ (2.0,3.0)\big\}$. These choices enable the systematic exploration of different uncertainty levels in the distance intervals. Notably, configurations with $\varepsilon_{\mathrm{short}} = 1\,\text{\AA}$ already provide a close approximation to experimentally derived data. In particular, fully experimental settings typically correspond to combinations in which $\varepsilon_{\mathrm{long}} \in \{2\,\text{\AA},\ 3\,\text{\AA}\}$, as discussed in \cite{2017_PSbN}.
	
	\medskip
	\noindent \textbf{Van der Waals lower bounds.}
	We considered the following van der Waals radii: $r_N = 1.55\,\text{\AA}$, $r_C = 1.70\,\text{\AA}$, and $r_H = 1.20\,\text{\AA}$ \cite{1964_bondi_vdwvar}. The natural lower bound between atoms $A_i$ and $A_j$ is given by the interval $\big[\underline{d}_{A_i,A_j}, \infty\big)$, where $\underline{d}_{A_i,A_j} = \mathrm{vdW}_t(A_i,A_j) = r_{A_i} + r_{A_j}$ (see Assumption~\ref{sec:4:vdwt}). However, for certain atom pairs in the PDB structures, we observed instances where the distance $d_{i,j}$ computed from the PDB coordinates is smaller than this value. To avoid excluding such configurations, we introduced an adaptive relaxation of this lower bound. Specifically, for two carbon atoms (and analogously for other atom types), we set
	\begin{equation*}
		\mathrm{vdW}_t(C_i,C_j) = \alpha \cdot (2r_C), \quad \text{where } \alpha = \max\big\{k \in \{1.0,\ 0.9,\ \dots,\ 0.1\} \mid d_{i,j} > k \cdot 2r_C \big\}.
		\label{sec:4:eq:vdWT}
	\end{equation*}
	
	The van der Waals threshold was applied selectively depending on the number of covalent bonds separating the atoms. No van der Waals constraint was enforced for atom pairs separated by one or two covalent bonds. For atom pairs separated by exactly three covalent bonds, we fixed $\alpha = 0.8$, except for backbone torsion angles ($\omega,\ \phi,\ \psi$), for which no van der Waals constraint was imposed. For atom pairs separated by at least four covalent bonds, the adaptive rule described above was employed. Additionally, for hydrogen--hydrogen pairs, a fixed relaxed threshold of $1.8\,\text{\AA}$ was used, as suggested in \cite{2017_PSbN}.
	
	\medskip
	\noindent \textbf{Torsion angle intervals.}
	Following Assumption~\ref{sec:4:torsion} in \Cref{sec:4:1}, torsion angle intervals $\mathcal{T}^i = \big[\underline{\tau}_i,\overline{\tau}_i\big]$ for $\phi_i$ and $\psi_i$ were generated similarly. Based on \cite{SEoPCSUaDGA_2019}, we take $\varepsilon_\tau = 50^\circ$. The bounds were defined as $\underline{\tau}_i = \tau_i^* - \varepsilon_{\tau_i}/2$ and $\overline{\tau}_i = \tau_i^* + \varepsilon_{\tau_i}/2$, where $\tau_i^* \sim \mathcal{N}\big(\tau_i, (\varepsilon_{\tau_i}/8)^2\big)$ and $\tau_i$ is the torsion angle measured from the atomic coordinates of the PDB structure. Since the terms in equation \eqref{sec:2:eq:costau} can be rearranged to express a distance as a function of the torsion angle, each angular interval $\mathcal{T}^i$ induces a corresponding interval distance between the atoms defining the torsion, and is therefore converted into a distance constraint. The periodicity of torsion angles is explicitly taken into account in this conversion. If an interval $\mathcal{T}^i$ crosses a branch boundary or the periodic boundary (i.e., $0^\circ$ or $\pm 180^\circ$, respectively), it is interpreted as a union of subintervals, and the corresponding distance interval is defined by the minimum and maximum values attained over this set, ensuring consistency with all admissible torsion angles. This avoids artificial tightening of the lower bound due to inconsistent handling of angular boundaries.
	
	Whenever chirality was known for certain atom quadruples, the corresponding torsion angle was assumed to be known for the associated clique. The torsion angles defined within extended peptide groups were also treated as known.
	
	\medskip
	\noindent \textbf{Implementation.}
	A Python~3.10 implementation corresponding to the procedures described above is publicly available at \hyperref{https://github.com/wdarocha/pdb-parser}{}{}{\texttt{https://github.com/wdarocha/pdb-parser}}. All instances used in this study were generated using software release~v0.1.0.
	
	\section{Computational Experiments and Analysis}\label{sec:5}
	
	Let $G = (V, E, \mathcal{D})$ be the graph associated with the \textit{i}DDGP instance. $X$ denotes a computed solution and $X^*$ a reference structure. Both $X$ and $X^*$ are assumed to be centered at their respective centroids. To assess the quality of the obtained solutions, we used three standard metrics: the Mean Distance Error (MDE), the Largest Distance Error (LDE) \cite{2022_bruglieri_MFSoDGC}, and the Root Mean Square Deviation (RMSD) \cite{2014_liberti_edgaa, 1994_Maiorov_SoRMSDiCTDSPG}, defined as follows:
	\begin{align*}
		\Delta_{v_i,v_j}(G,X) & :=  \max\big\{0,\ \underline{d}_{i,j} - \|x_i-x_j\|,\ \|x_i-x_j\| - \overline{d}_{i,j}\big\},\\[0.2cm]
		\text{MDE}(G,X) & := \displaystyle \dfrac{1}{|E|} \sum_{\{v_i,v_j\} \in E} \Delta_{v_i,v_j}(G,X),\\[0.2cm]
		\text{LDE}(G,X) & := \displaystyle \max_{\{v_i,v_j\} \in E} \Big\{\Delta_{v_i,v_j}(G,X)\Big\},\\[0.2cm]
		\text{RMSD}(X,X^*) & := \displaystyle \dfrac{1}{\sqrt{|V|}} \min_{Q \in SO(3)} \|X^*-XQ\|_F, \text{ where } \|\cdot\|_F \text{ denotes the Frobenius}\\
		&\text{norm and } SO(3) \text{ is the group of } 3 \times 3 \text{ rotation matrices with determi-}\\
		&\text{nant one}.
		\label{sec:5:eq:metrics}
	\end{align*}
	
	The MDE and LDE metrics evaluate how well the solution satisfies the input instance constraints. On the other hand, RMSD measures the structural similarity between the computed conformation and the reference structure $X^*$. Thus, the RMSD values reported here are computed after optimal proper rotational alignment and do not identify a conformation with its mirror image. While no universal RMSD threshold exists for distinguishing protein structures, several studies have shown that values in the range of $2$--$4\,\text{\AA}$ typically correspond to meaningful structural differences rather than random fluctuations \cite{1994_Maiorov_SoRMSDiCTDSPG, 1986_ChothiaLesk_RbDSaSiP, WitPoaCPoaPSwaRMSDof6A_1998, 2001_CarugoPongor_AnRMSDfCPTDS}. In practice, thresholds around $2$--$3\,\text{\AA}$ are widely adopted in structural bioinformatics to discriminate similar from distinct conformations. In this work, we adopt a conservative cutoff of $3\,\text{\AA}$ to define structurally distinct embeddings.
	
	Algorithms~\ref{sec:2:pc:iBP} and~\ref{sec:3:pc:iTBP} were implemented in C and compiled using GCC version 13.3.0 with the \texttt{-O3} optimization flag. Experiments were conducted on a computer equipped with an Intel Xeon Silver 4114 CPU (10 cores, 20 threads, max frequency $3.0$ GHz) and $156$ GB of RAM, running Ubuntu~24.04.3 LTS. Each algorithm instance was executed using a single thread, with up to seven instances running in parallel. The CPU time limit for each run was set to twelve hours.
	For the computational experiments, Algorithms~\ref{sec:2:pc:iBP} and~\ref{sec:3:pc:iTBP}, which are written in first-realization form, were executed in enumeration form: after each feasible sampled realization was recorded, the same depth-first traversal was continued to compute the reported solution counts, representative counts, and RMSD statistics.
	
	Although a first feasible realization is typically obtained well before the time cutoff, structurally distinct solutions generally require exploring larger portions of the search tree. Moreover, for runs in which no feasible sampled realization is found, the 12-hour limit permits substantial exploration of the sampled search tree before termination due to the time constraint. The source code, accompanying documentation, and the dataset of benchmark instances are publicly available at \hyperref{https://github.com/wdarocha/benchmarks/tree/main/ibp-iabp-itbp}{}{}{\texttt{https://github.com/wdarocha/benchmarks/tree/main/ibp-iabp-itbp}}. All computational experiments were conducted using software version~v1.0.0.
	
	The PDB structures were selected based on the number of amino acid residues (${n_{\text{aa}}}$) they contain. We also define $E_H \subset E_I \subset E$ as the subset of edges in $G$ whose weights correspond to interval distances between hydrogen atoms, with both bounds known, unlike van der Waals distances, which have only a lower bound. This set provides a quantitative measure of how spatially folded the protein structure is, as it is directly related to the set of adjacent predecessors of each vertex associated with a hydrogen atom. By a straightforward counting argument, based on~\eqref{sec:4:eq:order_hc} and \autoref{sec:4:fig:DDGP_hcOrder}, we obtain
	\begin{equation*}
		|E_H| = (\kappa_6 - 2) + \sum_{i=2}^{n_{\text{aa}}} \big[(\kappa_{5i-2} - 2) + (\kappa_{5i + 1} - 2)\big] = \kappa_6 + \sum_{i=2}^{n_{\text{aa}}} \big(\kappa_{5i-2} + \kappa_{5i + 1}\big) + 2 - 4n_{\text{aa}},
	\end{equation*}
	
	\noindent where $\kappa_i = |U_{v_i}|$. This set characterizes the collection of candidates responsible for branching and pruning operations in the three algorithms. More precisely, $2n_{\text{aa}} - 1$ elements are associated with the branching phase, while the remaining elements correspond to pruning operations in \textit{i}BP and to arc-reduction steps in \textit{i}ABP and \textit{i}TBP. \autoref{sec:5:tab:instances} presents the dataset used.
	
\begin{table}[!htp]
	\def\arraystretch{1.3}
	\centering
	\caption{Summary of the protein dataset used in the experiments.}
	\label{sec:5:tab:instances}
	\begin{tabular}{@{}cccccc|@{}|cccccc@{}}
		\toprule
		{PDB id} & ${n_{\text{aa}}}$ & {$|V|$} & {$|E_0|$} & {$|E_I|$} & {$|E_H|$} & {PDB id} & ${n_{\text{aa}}}$ & {$|V|$} & {$|E_0|$} & {$|E_I|$} & {$|E_H|$}\\
		\thickhline
		\texttt{1tos} & & & & & $56$ & \texttt{1ho0} & & & & & $167$\\
		\texttt{1uao} & $10$ & $52$ & $141$ & $1{,}185$ & $70$ & \texttt{1mmc} & $30$ & $152$ & $421$ & $11{,}055$ & $223$ \\
		\texttt{1kuw} & & & & & $76$ & \texttt{1d0r} & & & & & $229$\\
		\hline
		\texttt{1id6} & & & & & $106$ & \texttt{1zwd} & & & & & $240$\\
		\texttt{1dng} & $15$ & $77$ & $211$ & $2{,}715$ & $114$ & 1d1h & $35$ & $177$ & $491$ & $15{,}085$ & $245$ \\
		\texttt{1o53} & & & & & $116$ & \texttt{1spf} & & & & & $277$\\
		\hline
		\texttt{1du1} & & & & & $128$ & \texttt{1aml} & & & & & $267$\\
		\texttt{1dpk} & $20$ & $102$ & $281$ & $4{,}870$ & $138$ & \texttt{1ba4} & $40$ & $202$ & $561$ & $19{,}740$ & $301$\\
		\texttt{1ho7} & & & & & $166$ & \texttt{1c56} & & & & & $303$\\
		\hline
		\texttt{1ckz} & & & & & $148$ & & & & & & \\
		\texttt{1lfc} & $25$ & $127$ & $351$ & $7{,}650$ & $163$ & & & & & & \\
		\texttt{1a11} & & & & & $207$ & & & & & & \\
		\bottomrule
	\end{tabular}
\end{table}
	
	The analysis in this section focuses on the representative interval configuration $(\varepsilon_{\mathrm{short}}, \varepsilon_{\mathrm{long}}) = (1.0,2.0)$, which yields a moderate and experimentally meaningful level of uncertainty in the distance intervals. This choice is consistent with parameter ranges commonly used to approximate NMR-derived data. Additional interval configurations were also evaluated and are discussed at the end of the section.
	
	For this configuration, we first conducted a benchmark to determine the angular sampling size used by the three algorithms (\textit{i}BP, \textit{i}ABP, and \textit{i}TBP). For each vertex, the feasible angular set $\mathcal{T}_i$ is stored as negative and positive angular parts, $\mathcal{T}_i^-$ and $\mathcal{T}_i^+$, corresponding respectively to angles in $(-\pi,0]$ and $[0,\pi]$. The parameters $\nu_i^-$ and $\nu_i^+$ specify the maximum numbers of sample points requested from $\mathcal{T}_i^-$ and $\mathcal{T}_i^+$, respectively. Each angular part is represented as an ordered finite union of interval components, with overlapping or numerically adjacent components merged. Components crossing the angular cut are represented circularly so that the seam point is counted only once. The operator $\operatorname{sample}(\mathcal{T}_i^\pm,\nu_i^\pm)$ samples $\mathcal{T}_i^-$ and $\mathcal{T}_i^+$ separately and returns $T_i=T_i^-\cup T_i^+\subset\mathcal{T}_i$, where $|T_i^\pm|\leq\nu_i^\pm$ and hence $|T_i| \leq \nu_i = \nu_i^-+\nu_i^+$. Within each part, the points are sampled uniformly with respect to interval length: if $\ell_i^\pm$ is the total angular interval length of $\mathcal{T}_i^\pm$, the sampling positions are placed uniformly on $[0,\ell_i^\pm]$ and then mapped back to the corresponding component. If the prescribed minimum spacing of $5^\circ$ cannot be maintained, the number of sampled angles is reduced accordingly; no connected component of $\mathcal{T}_i^\pm$ is discarded before sampling. \autoref{sec:5:tab:sampesize} reports, as percentages over the 21 test instances, the proportion of cases in which each algorithm exhausted the sampled search tree, reached the 12-hour CPU time limit, or found at least one feasible realization. All runs in this calibration experiment used the same 12-hour time limit adopted in the main computational comparison. We restricted the sampling sizes to odd numbers so that, for a single interval component and inactive spacing constraint, the uniform sample contains both endpoints and a central point. For angular parts with multiple components, this endpoint--center property is not imposed componentwise; the same interval-length sampling rule over the ordered union is used. For each method, the selected sampling size was the smallest value attaining the largest percentage of feasible realizations. This tie-breaking rule avoids unnecessary enlargement of the search tree; in addition, our preliminary tests indicated that, once the maximum percentage of feasible realizations is reached, increasing the sampling size does not systematically improve the geometric variety of the conformations obtained.
	
\begin{table}[!htp]
	\def\arraystretch{1.3}
	\centering
	\caption{Qualitative results obtained by varying the maximum number of sample points requested from each angular part ($\nu_i^\pm$) for the three algorithms (\textit{i}BP, \textit{i}ABP and \textit{i}TBP). All entries are reported as percentages (\%) over the 21 test instances. The labels ``exh.'', ``timeout'', and ``feas.'' are mutually exclusive and denote, respectively, the percentage of instances for which the sampled search tree was exhausted without finding a feasible realization, the 12-hour CPU time limit was reached without finding a feasible realization, or at least one feasible realization was found. Once a feasible realization is found, the instance is classified under ``feas.'' regardless of whether the subsequent enumeration terminates by exhaustion or by timeout. For each algorithm, the entries corresponding to the smallest value of $\nu_i^\pm$ yielding the best performance are highlighted in bold.}
	\label{sec:5:tab:sampesize}
	\begin{tabular}{@{}cc|c|cc|c|cc|c|c@{}}
		\toprule
		&
		\multicolumn{3}{c}{\textit{i}BP} &
		\multicolumn{3}{c}{\textit{i}ABP} &
		\multicolumn{3}{c}{\textit{i}TBP} \\
		\cmidrule(lr){2-4} \cmidrule(lr){5-7} \cmidrule(lr){8-10}
		$\nu_i^\pm$ & 
			exh. & timeout & feas. & 
			exh. & timeout & feas. & 
			exh. & timeout & feas. \\
			\thickhline
			3 & 100.00 & 0.00 & 0.00 & 52.38 & 19.05 & 28.57 & 85.71 & 4.76 & 9.52 \\
			\textbf{5} & 61.90 & 28.57 & 9.52 & \textbf{0.00} & \textbf{19.05} & \textbf{80.95} & \textbf{0.00} & \textbf{9.52} & \textbf{90.48} \\
			7 & 14.29 & 38.10 & 47.62 & 0.00 & 23.81 & 76.19 & 0.00 & 19.05 & 80.95 \\
			9 & 0.00 & 28.57 & 71.43 & 0.00 & 23.81 & 76.19 & 0.00 & 23.81 & 76.19 \\
			\textbf{11} & \textbf{0.00} & \textbf{23.81} & \textbf{76.19} & 0.00 & 23.81 & 76.19 & 0.00 & 23.81 & 76.19 \\
			13 & 0.00 & 23.81 & 76.19 & 0.00 & 23.81 & 76.19 & 0.00 & 23.81 & 76.19 \\
			15 & 0.00 & 23.81 & 76.19 & 0.00 & 23.81 & 76.19 & 0.00 & 23.81 & 76.19 \\
			17 & 0.00 & 23.81 & 76.19 & 0.00 & 23.81 & 76.19 & 0.00 & 23.81 & 76.19 \\
			19 & 0.00 & 23.81 & 76.19 & 0.00 & 23.81 & 76.19 & 0.00 & 23.81 & 76.19 \\
			21 & 0.00 & 23.81 & 76.19 & 0.00 & 23.81 & 76.19 & 0.00 & 23.81 & 76.19 \\
		\bottomrule
	\end{tabular}
\end{table}
		
	Based on these results, we adopted $\nu_i^\pm=5$ for \textit{i}ABP and \textit{i}TBP. For \textit{i}BP, the same sampling size led to markedly inferior performance, and we therefore adopted $\nu_i^\pm=11$. Thus, the main experiments request at most 22 sampled angles for \textit{i}BP and at most 10 sampled angles for \textit{i}ABP and \textit{i}TBP at each branching step, before possible reductions imposed by the $5^\circ$ spacing rule. In all cases, the sample points were selected in increasing order within each angular part, starting with the negative part. This convention is motivated by the observation that negative values of torsion angles are more frequently populated in protein structures~\cite{2023_PDAiPPaNDotRM}.

	The results in \autoref{sec:5:tab:run_output} report only the instances for which at least one of the three methods did not find a feasible realization within the CPU time limit. For all remaining instances, each method found at least one feasible realization within the prescribed time bound. Instances not listed in the table also ran for twelve hours but yielded at least one valid solution before termination.
		
	At this level of comparison, \textit{i}ABP provides only a limited improvement over \textit{i}BP, finding a feasible realization for one additional instance. For the instances in which at least one of these two methods does not find a feasible realization, the fractions of embedded vertices do not indicate a conclusive advantage for either \textit{i}BP or \textit{i}ABP. In contrast, \textit{i}TBP finds feasible realizations for two instances in which neither \textit{i}BP nor \textit{i}ABP finds one, indicating a clear robustness gain from incorporating torsion-sign information.
		
\begin{table}[!htp]
	\def\arraystretch{1.3}
	\centering
	\caption{Numerical results after the 12-hour CPU time limit, reported as the fraction of the last embedded vertex over the total number of vertices when no feasible realization is found. A ``$\checkmark$'' indicates that the method found at least one feasible realization within this time limit.}
	\label{sec:5:tab:run_output}
	\begin{tabular}{@{}ccccccc@{}}
		\toprule
		{PDB id} & & \multicolumn{1}{c}{\textit{i}BP ($\nu_i^\pm = 11$)}
		& & \multicolumn{1}{c}{\textit{i}ABP ($\nu_i^\pm = 5$)}
		& & \multicolumn{1}{c}{\textit{i}TBP ($\nu_i^\pm = 5$)}
		\\
		\midrule
		\texttt{1lfc} & & $104/127$ & & \checkmark & & \checkmark \\
		\hline
		\texttt{1mmc} & & $\phantom{0}97/152$ & & $\phantom{0}97/152$ & & $97/152$ \\
		\hline
		\texttt{1d1h} & & $\phantom{0}70/177$ & & $\phantom{0}72/177$ & & \checkmark \\
		\hline
		\texttt{1ba4} & & $100/202$ & & $100/202$ & & \checkmark \\
		\texttt{1c56} & & $177/202$ & & $155/202$ & & $180/202$ \\
		\bottomrule
	\end{tabular}
\end{table}
	
	To analyze the output of each method, we introduce two complementary metrics that quantify the exploration effort and the embedding-normalized feasible-solution yield. To compare the exploration effort between two algorithms $a_1$ and $a_2$, we define
	\begin{equation}
		\xi^{a_1}_{a_2} = \frac{e_{a_1}}{e_{a_2}},
		\label{sec:5:eq:xi}
	\end{equation}
	
	\noindent where $e_a$ denotes the total number of vertices embedded in $\mathbb{R}^3$ by algorithm $a$. Thus, $\xi^{a_1}_{a_2}$ quantifies the relative embedding workload of $a_1$ and $a_2$. Values $\xi^{a_1}_{a_2} > 1$ indicate that $a_1$ performs more vertex embeddings than $a_2$, whereas $\xi^{a_1}_{a_2} < 1$ indicates the opposite.
	
	We then introduce an embedding-normalized feasible-solution yield metric that accounts for the minimum structural cost required to construct a feasible realization and enables comparison across instances of varying sizes. The metric is defined as
	\begin{equation}
		\eta_a = \frac{|V|\, s_a}{e_a},
		\label{sec:5:eq:eta}
	\end{equation}
	
	\noindent where $s_a$ denotes the number of feasible solutions found. Since $e_a/|V|$ represents the number of equivalent full constructions performed during the search, $\eta_a$ measures the number of feasible solutions obtained per equivalent full construction. By construction, $\eta_a=0$ if no feasible solution is found, while $\eta_a>1$ indicates that multiple realizations are obtained after amortizing the cost of a full construction.
	
	Finally, to assess structural diversity and geometric accuracy, we compute the number of considered solutions $c_a$, defined as the number of realizations whose pairwise proper-rotation RMSD (between solutions) is at least $3\,\text{\AA}$. The representative set counted by $c_a$ is constructed sequentially by a greedy, order-dependent rule: each newly found feasible realization is compared with all representatives already stored, and it is retained only if its proper-rotation RMSD from every stored representative is at least $3\,\text{\AA}$. We also compute the maximum MDE ($\overline{\mathrm{MDE}}$), maximum LDE ($\overline{\mathrm{LDE}}$), and minimum RMSD ($\underline{\mathrm{RMSD}}$) with respect to the reference structure from the original PDB file used to generate the instance; the result tables report $\underline{\mathrm{RMSD}}$, while the largest observed values of $\overline{\mathrm{MDE}}$ and $\overline{\mathrm{LDE}}$ are described in the text below. Unlike $c_a$, these three metrics are updated over every feasible realization generated during the search, regardless of whether that realization is retained as a structurally distinct representative.
	
	The algorithms employ different feasibility mechanisms and therefore require different numerical tolerances. While \textit{i}BP directly tests the distance constraints, \textit{i}ABP and \textit{i}TBP use an angular tolerance in their interval-intersection operations. Preliminary experiments with \textit{i}ABP indicated that an angular tolerance of $0.5^\circ$ provided satisfactory numerical performance; the same value was retained for \textit{i}TBP, which uses the same arc-reduction mechanism. Because \textit{i}BP and \textit{i}ABP rely on the same input information, the distance tolerance for \textit{i}BP was calibrated against the errors obtained by \textit{i}ABP. Specifically, since the largest observed $\overline{\mathrm{LDE}}$ for \textit{i}ABP was $1.81\times10^{-2}\,\text{\AA}$, the tolerance for \textit{i}BP was set to a slightly larger value at the adopted precision, $1.90\times10^{-2}\,\text{\AA}$.
		
	As expected from its direct distance-feasibility test, the largest $\overline{\mathrm{LDE}}$ obtained by \textit{i}BP is approximately $1.90\times10^{-2}\,\text{\AA}$ in each instance for which it finds a feasible realization. For \textit{i}TBP, the corresponding values remain at or below this reference in 18 of the 19 instances for which a feasible realization is found. In the remaining instance, the reference is exceeded by $3.30\times10^{-3}\,\text{\AA}$. The largest observed $\overline{\mathrm{MDE}}$ is $1.34\times10^{-4}\,\text{\AA}$ for \textit{i}BP and at most $3.06\times10^{-5}\,\text{\AA}$ for the angular methods. These values are consistent with the accuracy levels commonly reported in the DGP literature on protein structure determination; see, e.g.,~\cite{2017_dambrosio_NEMaMfRPGfDG,2018_liberti_BNAiDG,2022_bruglieri_MFSoDGC}.
	
	Among the three algorithms, the most meaningful comparison is between \textit{i}BP and \textit{i}ABP, as they rely on the same amount of input information. Since \textit{i}TBP integrates additional information while retaining the same arc-reduction strategy as \textit{i}ABP, we also perform a separate comparison between these two methods. \autoref{sec:5:tab:run_output_ibp_vs_iabp} reports the run-output metrics for the first comparison, whereas \autoref{sec:5:tab:run_output_iabp_vs_itbp} reports the corresponding metrics for the second. The results of both tables are discussed below.
	
\begin{table}[!htp]
	\def\arraystretch{1.2}
	\setlength{\tabcolsep}{6.3pt}
	\centering
	\caption{Numerical results for the run-output metrics $\xi^{a_1}_{a_2}$ and $\eta_a$, defined in \eqref{sec:5:eq:xi}--\eqref{sec:5:eq:eta}, along with the number of feasible solutions $s_a$, the number of distinct solutions $c_a$, and the minimum RMSD $\underline{\mathrm{RMSD}}$ over all solutions obtained by \textit{i}BP and \textit{i}ABP.}
	\label{sec:5:tab:run_output_ibp_vs_iabp}
	\begin{tabular}{@{}c cc|cc|cc|cc|c@{}}
		\toprule
		\multirow{2}{*}{\raisebox{-0.5\height}{PDB id}} &
		\multicolumn{1}{c}{$\xi^{a_1}_{a_2}$} &
		\multicolumn{2}{c}{$\eta_a$} &
		\multicolumn{2}{c}{$s_a \ (\times 10^6)$} &
		\multicolumn{2}{c}{$c_a$} &
		\multicolumn{2}{c}{\underline{RMSD} ($\text{\AA}$)} \\
		\cmidrule(lr){2-2} \cmidrule(lr){3-4} \cmidrule(lr){5-6} \cmidrule(lr){7-8} \cmidrule(lr){9-10}
		&
		{\footnotesize\phantom{.}\textit{i}BP/\textit{i}ABP\phantom{.}} &
		{\footnotesize\textit{i}BP} & {\footnotesize\textit{i}ABP} &
		{\footnotesize\textit{i}BP} & {\footnotesize\textit{i}ABP} &
		{\footnotesize\textit{i}BP} & {\footnotesize\textit{i}ABP} &
		{\footnotesize\textit{i}BP} & {\footnotesize\textit{i}ABP} \\
		\thickhline
		\texttt{1tos} & 9.70 & 2.81 & \textbf{11.88} & \textbf{2944.22} & 1283.01 & 4 & \textbf{5} & 2.72 & \textbf{2.33} \\
		\texttt{1uao} & 2.91 & 4.21 & \textbf{16.05} & 2069.61 & \textbf{2710.70} & 1 & \textbf{2} & 0.44 & \textbf{0.28} \\
		\texttt{1kuw} & 3.57 & 5.83 & \textbf{22.14} & 2182.68 & \textbf{2322.56} & 1 & 1 & 1.69 & \textbf{1.66} \\
		\hline
		\texttt{1id6} & 1.41 & 13.43 & \textbf{31.76} & 927.15 & \textbf{1550.74} & \textbf{2} & 1 & 3.87 & \textbf{3.84} \\
		\texttt{1dng} & 2.94 & 6.82 & \textbf{23.79} & 1327.48 & \textbf{1576.50} & 1 & 1 & 0.99 & \textbf{0.67} \\
		\texttt{1o53} & 1.44 & 17.23 & \textbf{28.41} & 1278.68 & \textbf{1461.96} & 2 & 2 & \textbf{0.72} & 0.87 \\
		\hline
		\texttt{1du1} & 1.89 & 16.45 & \textbf{24.22} & \textbf{948.50} & 740.78 & 3 & 3 & 2.66 & \textbf{2.52} \\
		\texttt{1dpk} & 2.13 & 23.93 & \textbf{43.61} & \textbf{931.16} & 796.17 & 2 & \textbf{5} & \textbf{3.18} & 3.52 \\
		\texttt{1ho7} & 1.49 & 21.85 & \textbf{29.20} & \textbf{1000.04} & 898.85 & 1 & 1 & 1.85 & \textbf{1.55} \\
		\hline
		\texttt{1ckz} & 1.31 & 32.92 & \textbf{40.49} & \textbf{642.04} & 603.93 & 1 & 1 & \textbf{6.00} & 9.66 \\
		\texttt{1lfc} & 28.85 & 0.00 & 0.004 & 0.00 & \textbf{0.85} & 0 & \textbf{1} & --- & \textbf{2.87} \\
		\texttt{1a11} & 2.13 & 20.91 & \textbf{53.28} & 575.73 & \textbf{690.03} & 1 & 1 & 1.41 & \textbf{1.33} \\
		\hline
		\texttt{1ho0} & 2.57 & 17.70 & \textbf{45.84} & 450.13 & \textbf{453.19} & 2 & \textbf{4} & 8.68 & \textbf{7.14} \\
		\texttt{1mmc} & 24.98 & 0.00 & 0.00 & 0.00 & 0.00 & 0 & 0 & --- & --- \\
		\texttt{1d0r} & 2.49 & 23.61 & \textbf{57.79} & \textbf{463.74} & 455.26 & 1 & 1 & 2.01 & \textbf{1.37} \\
		\hline
		\texttt{1zwd} & 2.13 & 39.43 & \textbf{69.98} & \textbf{388.95} & 324.77 & 1 & \textbf{3} & \textbf{4.70} & 6.08 \\
		\texttt{1d1h} & 42.53 & 0.00 & 0.00 & 0.00 & 0.00 & 0 & 0 & --- & --- \\
		\texttt{1spf} & 6.50 & 4.08 & \textbf{25.55} & \textbf{291.62} & 280.81 & 1 & 1 & 4.44 & \textbf{3.20} \\
		\hline
		\texttt{1aml} & 1.84 & 37.91 & \textbf{69.18} & \textbf{307.91} & 305.95 & 1 & 1 & 7.33 & \textbf{7.29} \\
		\texttt{1ba4} & 36.92 & 0.00 & 0.00 & 0.00 & 0.00 & 0 & 0 & --- & --- \\
		\texttt{1c56} & 34.60 & 0.00 & 0.00 & 0.00 & 0.00 & 0 & 0 & --- & --- \\
		\bottomrule
	\end{tabular}
\end{table}
	
\begin{table}[!htp]
	\def\arraystretch{1.2}
	\setlength{\tabcolsep}{6.3pt}
	\centering
	\caption{Numerical results for the run-output metrics $\xi^{a_1}_{a_2}$ and $\eta_a$, defined in \eqref{sec:5:eq:xi}--\eqref{sec:5:eq:eta}, along with the number of feasible solutions $s_a$, the number of distinct solutions $c_a$, and the minimum RMSD $\underline{\mathrm{RMSD}}$ over all solutions obtained by \textit{i}ABP and \textit{i}TBP. Instances marked with $*$ indicate that the entire sampled search tree was explored by \textit{i}TBP before reaching the time limit.}
	\label{sec:5:tab:run_output_iabp_vs_itbp}
	\begin{tabular}{@{}c c c|c c|c c|c c|c@{}}
		\toprule
		\multirow{2}{*}{\raisebox{-0.5\height}{PDB id}} &
		\multicolumn{1}{c}{$\xi^{a_1}_{a_2}$} &
		\multicolumn{2}{c}{$\eta_a$} &
		\multicolumn{2}{c}{$s_a \ (\times 10^6)$} &
		\multicolumn{2}{c}{$c_a$} &
		\multicolumn{2}{c}{\underline{RMSD} ($\text{\AA}$)} \\
		\cmidrule(lr){2-2} \cmidrule(lr){3-4} \cmidrule(lr){5-6} \cmidrule(lr){7-8} \cmidrule(lr){9-10}
		&
		{\footnotesize\textit{i}ABP/\textit{i}TBP} &
		{\footnotesize\textit{i}ABP} & {\footnotesize\textit{i}TBP} &
		{\footnotesize\textit{i}ABP} & {\footnotesize\textit{i}TBP} &
		{\footnotesize\textit{i}ABP} & {\footnotesize\textit{i}TBP} &
		{\footnotesize\textit{i}ABP} & {\footnotesize\textit{i}TBP} \\
		\thickhline
		\texttt{1tos}\phantom{$^*$} & 1.01 & \textbf{11.88} & 8.78 & \textbf{1283.01} & 934.99 & \textbf{5} & 2 & 2.33 & \textbf{0.18}\\
		\texttt{1uao}$^*$ & 512.43 & \textbf{16.05} & 3.23 & \textbf{2710.70} & 1.07 & \textbf{2} & 1 & 0.28 & \textbf{0.24}\\
		\texttt{1kuw}$^*$ & 52.70 & \textbf{22.14} & 11.34 & \textbf{2322.56} & 22.58 & 1 & 1 & 1.66 & \textbf{0.18}\\
		\hline
		\texttt{1id6}\phantom{$^*$} & 0.63 & \textbf{31.76} & 16.18 & \textbf{1550.74} & 1250.21 & 1 & 1 & 3.84 & \textbf{0.23}\\
		\texttt{1dng}$^*$ & 9.47 & \textbf{23.79} & 11.15 & \textbf{1576.50} & 77.98 & 1 & 1 & 0.67 & \textbf{0.25}\\
		\texttt{1o53}\phantom{$^*$} & 0.84 & \textbf{28.41} & 16.79 & \textbf{1461.96} & 1027.23 & \textbf{2} & 1 & 0.87 & \textbf{0.26}\\
		\hline
		\texttt{1du1}\phantom{$^*$} & 0.43 & \textbf{24.22} & 11.38 & 740.78 & \textbf{801.52} & 3 & \textbf{8} & 2.52 & \textbf{1.55}\\
		\texttt{1dpk}\phantom{$^*$} & 0.47 & \textbf{43.61} & 23.56 & 796.17 & \textbf{911.97} & \textbf{5} & 3 & 3.52 & \textbf{0.99}\\
		\texttt{1ho7}\phantom{$^*$} & 0.53 & \textbf{29.20} & 14.15 & \textbf{898.85} & 829.06 & 1 & 1 & 1.55 & \textbf{0.78}\\
		\hline
		\texttt{1ckz}\phantom{$^*$} & 0.40 & \textbf{40.49} & 17.97 & 603.93 & \textbf{669.92} & 1 & \textbf{5} & 9.66 & \textbf{3.38}\\
		\texttt{1lfc}\phantom{$^*$} & 1.05 & 0.004 & \textbf{0.56} & 0.85 & \textbf{100.49} & 1 & 1 & 2.87 & \textbf{0.50}\\
		\texttt{1a11}\phantom{$^*$} & 0.84 & \textbf{53.28} & 28.33 & \textbf{690.03} & 437.53 & 1 & 1 & 1.33 & \textbf{0.60}\\
		\hline
		\texttt{1ho0}\phantom{$^*$} & 0.33 & \textbf{45.84} & 13.94 & \textbf{453.19} & 419.63 & 4 & \textbf{11} & 7.14 & \textbf{2.84}\\
		\texttt{1mmc}\phantom{$^*$} & 0.77 & 0.00 & 0.00 & 0.00 & 0.00 & 0 & 0 & --- & ---\\
		\texttt{1d0r}\phantom{$^*$} & 0.60 & \textbf{57.79} & 30.69 & \textbf{455.26} & 400.59 & 1 & 1 & 1.37 & \textbf{1.05}\\
		\hline
		\texttt{1zwd}\phantom{$^*$} & 0.50 & \textbf{69.98} & 29.80 & \textbf{324.77} & 277.84 & 3 & 3 & 6.08 & \textbf{2.98}\\
		\texttt{1d1h}\phantom{$^*$} & 9.74 & 0.00 & \textbf{17.53} & 0.00 & \textbf{230.51} & 0 & \textbf{2} & --- & \textbf{0.41}\\
		\texttt{1spf}\phantom{$^*$} & 0.14 & \textbf{25.55} & 0.99 & \textbf{280.81} & 76.88 & 1 & \textbf{4} & 3.20 & \textbf{0.50}\\
		\hline
		\texttt{1aml}\phantom{$^*$} & 0.80 & \textbf{69.18} & 37.54 & \textbf{305.95} & 207.34 & 1 & \textbf{2} & 7.29 & \textbf{4.25}\\
		\texttt{1ba4}\phantom{$^*$} & 27.13 & 0.00 & \textbf{38.93} & 0.00 & \textbf{178.10} & 0 & \textbf{2} & --- & \textbf{1.25}\\
		\texttt{1c56}\phantom{$^*$} & 1.26 & 0.00 & 0.00 & 0.00 & 0.00 & 0 & 0 & --- & ---\\
		\bottomrule
	\end{tabular}
\end{table}
	
	\noindent \textbf{Overall remark.} The discrepancy between $s_a$ and $c_a$ reflects the distinction between feasibility and structural diversity. While $s_a$ counts all conformations satisfying the instance constraints (within numerical tolerance), $c_a$ denotes the number of structurally distinct solutions, where distinctness is defined by a minimum pairwise RMSD threshold of $3\,\text{\AA}$. Whenever the three algorithms successfully produce solutions, they typically generate them in huge numbers (in most cases reaching several hundred million); however, most of these conformations exhibit very small RMSD differences and are therefore structurally indistinguishable from a molecular standpoint.
	
	\medskip
	\noindent \textbf{\emph{i}BP vs.\ \emph{i}ABP.}
	The values of $\xi^{i\mathrm{BP}}_{i\mathrm{ABP}}$ range from $1.31$ to $42.53$ and exceed one for all instances. Thus, under the same time budget, \textit{i}BP always embeds more vertices than \textit{i}ABP. Averaged over the 21 instances, the embedding workload of \textit{i}BP is $10.21$ times that of \textit{i}ABP. The largest ratios should be interpreted together with the termination status of the methods, since they include instances for which no feasible realization is found by one or both algorithms within the prescribed time limit. Overall, this behavior is consistent with the algorithmic design: \textit{i}ABP spends more computation per embedding by performing arc-reduction operations (interval-intersections), whereas \textit{i}BP relies on comparatively inexpensive feasibility checks.
		
	The embedding-normalized yield metric $\eta_a$ provides a complementary view of this trade-off. Among the 16 instances for which both methods find feasible realizations, $\eta_{i\mathrm{ABP}} > \eta_{i\mathrm{BP}}$ in every case, indicating that \textit{i}ABP produces more feasible realizations per equivalent full construction. On average, the value for \textit{i}ABP is about $2.67$ times the corresponding value for \textit{i}BP. Hence, although \textit{i}ABP embeds fewer vertices, it attains a higher feasible-solution yield per equivalent full construction.
		
	The raw number of feasible realizations $s_a$ follows a more balanced pattern. Among the 16 instances for which both methods find feasible realizations, \textit{i}BP produces more feasible realizations in nine cases, whereas \textit{i}ABP does so in seven. Thus, the larger embedding workload of \textit{i}BP does not translate uniformly into a larger number of feasible realizations. Moreover, \textit{i}ABP finds a feasible realization for one additional instance for which \textit{i}BP does not.
		
	A clearer advantage for \textit{i}ABP appears in the number of structurally distinct realizations. On the 16 instances for which both methods find feasible realizations, \textit{i}ABP produces 33 distinct realizations, compared with 25 for \textit{i}BP; this corresponds to a $32.0\%$ increase. When all instances with at least one feasible realization are counted for each method, the corresponding totals are 34 and 25, giving a $36.0\%$ increase in favor of \textit{i}ABP.
		
	The minimum RMSD values $\underline{\mathrm{RMSD}}$ indicate an accuracy advantage for \textit{i}ABP. Among the 16 instances in which both methods produce feasible realizations, \textit{i}ABP attains the smaller RMSD in 12 cases, whereas \textit{i}BP does so in four. Therefore, while the geometric advantage is not uniform across all instances, \textit{i}ABP more frequently identifies the realization that best approximates the reference structure.
		
	\medskip
	\noindent \textbf{\emph{i}ABP vs.\ \emph{i}TBP.}
	The values of $\xi^{i\mathrm{ABP}}_{i\mathrm{TBP}}$ range from $0.14$ to $512.43$, revealing a highly skewed distribution. In 13 of the 21 instances, $\xi^{i\mathrm{ABP}}_{i\mathrm{TBP}} < 1$, meaning that \textit{i}ABP embeds fewer vertices than \textit{i}TBP in the majority of cases. Among the eight instances for which $\xi^{i\mathrm{ABP}}_{i\mathrm{TBP}} > 1$, three correspond to cases in which \textit{i}TBP exhaustively explores its search tree before the time limit, and two to cases in which \textit{i}TBP finds feasible realizations while \textit{i}ABP does not. In two further cases, both methods find feasible realizations, but \textit{i}TBP performs fewer embeddings. Finally, neither method finds a feasible realization for two instances; \textit{i}ABP performs fewer embeddings in one of them and \textit{i}TBP in the other. Thus, $\xi^{i\mathrm{ABP}}_{i\mathrm{TBP}}$ should be read primarily as an embedding-workload indicator. On the instances where the comparison is meaningful in terms of solution quality, the advantage of \textit{i}TBP appears not through a uniformly smaller embedding count, but through the structural metrics discussed below.
		
	The embedding-normalized yield metric $\eta_a$ clearly favors \textit{i}ABP on the instances where both methods find feasible realizations. Among these 17 instances, \textit{i}ABP has the larger value of $\eta_a$ in 16 cases. In the only exception, \textit{i}TBP produces substantially more feasible realizations relative to the corresponding embedding workload. Over the same set of instances, the value for \textit{i}ABP is about $3.49$ times the corresponding value for \textit{i}TBP, on average. The remaining losses of \textit{i}ABP in $\eta_a$ occur in the two instances where \textit{i}ABP finds no feasible realization while \textit{i}TBP succeeds.
		
	Concerning the raw number of feasible realizations $s_a$, \textit{i}ABP returns more feasible solutions in 13 of the 17 instances for which both methods find feasible realizations, whereas \textit{i}TBP does so in four. This tendency is consistent with the less constrained search space explored by \textit{i}ABP, although the four exceptions show that it does not hold uniformly. In particular, for the instances where \textit{i}TBP exhaustively explores the search tree before the time limit, the smaller values of $s_{i\mathrm{TBP}}$ reflect the complete enumeration of a more restricted admissible space. Thus, larger values of $s_{i\mathrm{ABP}}$ do not necessarily imply a larger set of biologically meaningful conformations; they mainly indicate that more distance-feasible realizations are available when torsion-sign restrictions are not imposed. Conversely, \textit{i}TBP finds feasible realizations for two instances in which \textit{i}ABP finds none, while for two instances neither method finds a feasible realization within the time limit. This is why the subsequent metrics $c_a$ and $\underline{\mathrm{RMSD}}$ are more informative for assessing structural relevance.
		
	A clearer distinction emerges in the number of structurally distinct realizations. \textit{i}TBP finds feasible realizations for 19 of the 21 instances, compared with 17 for \textit{i}ABP. On the 17 instances for which both methods find feasible realizations, \textit{i}TBP produces 47 distinct realizations, compared with 34 for \textit{i}ABP; this corresponds to a $38.2\%$ increase. When all instances with at least one feasible realization are counted for each method, the corresponding totals are 51 and 34, giving a $50.0\%$ increase in favor of \textit{i}TBP. These results indicate that the additional torsion-sign information guides the search toward a more structurally differentiated set of feasible conformations.
		
	This structural advantage is also evident in the minimum RMSD values $\underline{\mathrm{RMSD}}$. In all 17 instances where both methods produce feasible realizations, \textit{i}TBP attains the smallest RMSD, often by a substantial margin. Hence, among all realizations generated during the search, \textit{i}TBP more consistently identifies the configuration that most closely reproduces the reference structure used to define the instance. This behavior is consistent with the effect of incorporating torsion-angle information, which steers the search toward conformations compatible with physically meaningful local geometry.
		
	Overall, the distinguishing feature of \textit{i}TBP lies in its geometrically informed search mechanism. By integrating torsion angle sign constraints into the construction phase, it reshapes the combinatorial structure of the search tree and tends to favor geometrically faithful and structurally coherent realizations. This suggests that \textit{i}TBP is especially useful when fidelity to realistic molecular geometry is a primary objective.
	
	\medskip
	\noindent \textbf{Summary.}
	The three methods exhibit distinct computational regimes: \textit{i}BP emphasizes breadth of exploration through inexpensive vertex embeddings; \textit{i}ABP improves the embedding-normalized feasible-solution yield by integrating interval-based pruning mechanisms; and \textit{i}TBP further restructures the effective search space by incorporating angular constraints, thereby promoting geometrically faithful and structurally coherent realizations. Together, these results illustrate the fundamental trade-off between exploration breadth, pruning effectiveness, and geometric guidance within the \textit{i}DDGP framework. Although there are two instances for which \textit{i}TBP does not find a feasible realization within the time limit, the reported metrics indicate a clear benefit from torsion-angle information: \textit{i}TBP finds feasible realizations for more instances than \textit{i}ABP, produces more structurally distinct realizations when all instances with feasible realizations are counted, and attains the lowest minimum RMSD on the instances for which both methods find feasible realizations.
	
	\medskip
	\noindent \textbf{Additional interval configurations.}
	To complement the analysis above, we evaluated the remaining interval pairs $(\varepsilon_{\mathrm{short}}, \varepsilon_{\mathrm{long}})$ (see \Cref{subsec:instance_generation}) on a reduced dataset consisting of one representative instance from each of the seven size classes listed in \autoref{sec:5:tab:instances}, totaling one third of the full dataset (instances \texttt{1kuw}, \texttt{1o53}, \texttt{1ho7}, \texttt{1a11}, \texttt{1ho0}, \texttt{1spf}, and \texttt{1aml}). This subset contains only instances for which all three methods found feasible realizations in the main configuration. The maximum number of sample points requested from each angular part was determined according to the same procedure described previously, but applied to this subset. All experiments were conducted under the same computational settings as in the previous analysis, including a time limit of 12 hours per instance. For the additional interval configurations, the $\overline{\mathrm{LDE}}$ values obtained by \textit{i}BP and \textit{i}ABP do not exceed $1.90\times10^{-2}\,\text{\AA}$. Only one of the 28 \textit{i}TBP runs exceeds this reference, by $2.70\times10^{-3}\,\text{\AA}$. The largest observed $\overline{\mathrm{MDE}}$ values are $1.30\times10^{-4}\,\text{\AA}$ for \textit{i}BP and $3.08\times10^{-5}\,\text{\AA}$ for the angular methods. This behavior is consistent with that observed for $(1.0,2.0)$. The structural metrics follow the same qualitative patterns as in that configuration; accordingly, we focus the presentation on the number of structurally distinct solutions $c_a$ and the minimum RMSD $\underline{\mathrm{RMSD}}$. The corresponding results are summarized in \autoref{sec:5:tab:other_results}, which reports, for each method (\textit{i}BP, \textit{i}ABP, and \textit{i}TBP), the medians $\mathrm{med}(c_a)$ and $\mathrm{med}(\underline{\mathrm{RMSD}})$ across the seven instances, followed in each case by the corresponding range $[\min,\max]$; the corresponding value of $\nu_i^\pm$ is also reported.
		
\begin{table}[!htp]
	\def\arraystretch{1.3}
	\setlength{\tabcolsep}{6.0pt}
	\centering
	\caption{
		Summary of results for additional interval configurations evaluated on a reduced dataset comprising instances \texttt{1kuw}, \texttt{1o53}, \texttt{1ho7}, \texttt{1a11}, \texttt{1ho0}, \texttt{1spf}, and \texttt{1aml}. The table reports the medians $\mathrm{med}(c_a)$ and $\mathrm{med}(\underline{\mathrm{RMSD}})$, each followed by the range $[\min,\max]$ computed across the seven instances, where $\underline{\mathrm{RMSD}}$ is measured over all feasible solutions produced by each method with respect to the reference structure from PDB. The maximum number of sample points requested from each angular part ($\nu_i^\pm$) is also reported. The interval parameters $(\varepsilon_{\mathrm{short}}, \varepsilon_{\mathrm{long}})$ are selected from $\big\{(0.1,0.5),\ (0.5,1.0),\ (1.0,3.0),\ (2.0,3.0)\big\}$.
	}
	\label{sec:5:tab:other_results}
	\begin{tabular}{@{}c c c c c c@{}}
		\toprule
		$\varepsilon_{\mathrm{short}}\,(\text{\AA})$ & 
		$\varepsilon_{\mathrm{long}}\,(\text{\AA})$ & 
		$a$ & 
		$\nu_i^\pm$ & 
		$\mathrm{med}(c_a),\ [\min,\max]$ & 
		$\mathrm{med}(\underline{\mathrm{RMSD}}),\ [\min,\max]\,(\text{\AA})$ \\
		\midrule
		
		\multirow{3}{*}{$0.1$} & \multirow{3}{*}{$0.5$}
		&   \textit{i}BP  & 5 & $1,\ [1,3]$ & $2.83,\ [0.07,7.11]$ \\
		& & \textit{i}ABP & 5 & $1,\ [1,2]$ & $2.84,\ [0.05,4.90]$ \\
		& & \textit{i}TBP & 5 & $1,\ [1,3]$ & $0.10,\ [0.06,2.72]$ \\
		\midrule
		
		\multirow{3}{*}{$0.5$} & \multirow{3}{*}{$1.0$}
		&   \textit{i}BP  & 7 & $1,\ [1,3]$ & $2.70,\ [0.11,5.75]$ \\
		& & \textit{i}ABP & 5 & $1,\ [1,2]$ & $2.42,\ [0.12,9.64]$ \\
		& & \textit{i}TBP & 5 & $1,\ [1,6]$ & $0.44,\ [0.14,4.38]$ \\
		\midrule
		
		\multirow{3}{*}{$1.0$} & \multirow{3}{*}{$3.0$}
		&   \textit{i}BP  & 11 & $1,\ [1,3]$ & $2.80,\ [0.34,7.71]$ \\
		& & \textit{i}ABP & 5  & $1,\ [1,5]$ & $1.57,\ [0.30,8.08]$ \\
		& & \textit{i}TBP & 5  & $1,\ [1,10]$ & $0.86,\ [0.27,9.54]$ \\
		\midrule
		
		\multirow{3}{*}{$2.0$} & \multirow{3}{*}{$3.0$}
		&   \textit{i}BP  & 11 & $1,\ [1,3]$ & $2.90,\ [0.69,8.27]$ \\
		& & \textit{i}ABP & 7  & $1,\ [1,3]$ & $2.04,\ [1.32,8.37]$ \\
		& & \textit{i}TBP & 7  & $1,\ [1,6]$ & $1.41,\ [0.15,8.30]$ \\
		\bottomrule
	\end{tabular}
\end{table}
		
	For all interval configurations, the median value of $c_a$ is equal to one for the three methods, indicating that a typical instance in this reduced subset still yields a small number of structurally distinct realizations. The ranges, however, reveal greater variation across instances as interval uncertainty changes. This variation is most pronounced for \textit{i}TBP, whose maximum observed value of $c_a$ reaches 10 in the $(1.0,3.0)$ setting and 6 in the $(2.0,3.0)$ setting. The $\underline{\mathrm{RMSD}}$ statistics show a clearer and more consistent separation. Across all four interval configurations, \textit{i}TBP attains the lowest median $\underline{\mathrm{RMSD}}$, with the strongest contrast occurring for the narrower intervals. The $\underline{\mathrm{RMSD}}$ ranges vary nonmonotonically across methods and interval configurations, reflecting instance-level heterogeneity, while the median values consistently favor the torsion-guided strategy. Overall, these experiments with additional interval configurations support the same conclusion as the previous analysis: torsion-sign information does not increase the median number of structurally distinct realizations, but can broaden its observed range while improving the typical agreement with the reference geometry.
		
	\section{Conclusions and Future Directions}\label{sec:6}
	
	Motivated by the ideas proposed in \cite{2021_lavor_itspitibpaftdmdgp} for the \textit{i}DMDGP, this study puts forward two novel algorithmic frameworks for solving the \textit{i}DDGP in $\mathbb{R}^3$: the \textit{i}ABP and the \textit{i}TBP. These algorithms exploit geometric properties of torsion angles to improve the discretization and pruning strategies of the classical \textit{i}BP method. The \textit{i}ABP reformulates distance intervals as angular intervals, enabling a more structured and constraint-aware sampling process. The \textit{i}TBP further extends this approach by incorporating torsion angle information that may be known \textit{a priori}, allowing the algorithm to enforce local chirality constraints that are critical in biomolecular contexts.
	
	We also proposed a systematic strategy to generate \textit{i}DDGP instances directly from PDB structures, establishing a consistent and biologically meaningful modeling pipeline. Our computational experiments indicate that \textit{i}ABP improves the embedding-normalized feasible-solution yield relative to \textit{i}BP and that, in the main comparison, both angular methods find feasible realizations for more instances than \textit{i}BP. Furthermore, we found that the robustness of the algorithms depends strongly on the angular sampling resolution. In terms of sampling requirements, \textit{i}BP generally demands more points per angular part than \textit{i}ABP and \textit{i}TBP. Under the adopted calibration procedure for the main interval-width configuration, $(\varepsilon_{\mathrm{short}},\varepsilon_{\mathrm{long}})=(1.0,2.0)$, the selected sampling size was $\nu_i^\pm=11$ for \textit{i}BP and $\nu_i^\pm=5$ for both angular methods. Tests with other interval widths show that this difference is not uniform, but becomes more pronounced for wider intervals. This indicates that, as interval uncertainty increases, \textit{i}BP requires denser angular sampling to compensate for the absence of the common-frame angular-set reduction used by the proposed methods. Nevertheless, further improvements are necessary, as some instances did not yield a feasible realization within the allotted time. In particular, the main comparison shows that \textit{i}TBP benefits from the additional torsion-angle information, finding feasible realizations for more instances than \textit{i}ABP and attaining lower minimum RMSD values on the instances for which both angular methods find feasible realizations. These findings indicate that torsion-angle-based discretization schemes are promising for controlled PDB-derived interval instances and motivate further investigation in more realistic structure-inference settings.
	
	As further directions, we note that the results are highly dependent on the chosen \textit{hc} order. Since alternative orderings are possible, determining an optimal vertex order remains an open problem. Another relevant aspect concerns the sampling of the feasible angular intervals, which is performed uniformly in the present study. A natural extension is to investigate data-driven sampling strategies that prioritize angular values according to their empirical frequency in protein structures, using statistical information extracted from the PDB, as in \cite{2024_APAitSotMDGP}. Such strategies may provide a more informative selection of candidate angles while preserving the interval-based framework. Moreover, the use of the cross product in this study implies that many of the geometric arguments are specific to $\mathbb{R}^3$. A natural extension of this work is to explore if the proposed framework can be generalized to Euclidean spaces $\mathbb{R}^K$ with $K > 3$.
	
	\subsection*{Acknowledgments}
	
	\FUNDING
	
	\section*{Conflicts of Interest} \par
	
	The authors declare no conflicts of interest.
	

\end{document}